\documentclass[journal,onecolumn]{IEEEtran}
\usepackage[latin1]{inputenc}
\usepackage{amsmath}
\usepackage{amsthm}
\usepackage{amsfonts}
\usepackage{amssymb}
\usepackage{graphicx}
\usepackage{tikz}
\usetikzlibrary{arrows}
\usepackage{pgfplots}
\usepackage{amsfonts}
\usepackage{amssymb}
\usepackage{graphicx}
\usepackage{mathrsfs}
\usepackage{tabularx,ragged2e}
\usepackage{hyperref}
\usepackage{color}
\usepackage{tikz}
\usepackage{enumerate}
\usepackage{cancel}
\usepackage{pdfpages}
\usepackage{tabularx,ragged2e}
\usepackage{cleveref}
\usepackage{booktabs}
\usepackage{algorithm}
\usepackage{algpseudocode}
\usepackage{mdframed}
\usepackage{mathtools}
\usepackage{xparse}
\usepackage[font=small,labelfont=bf,tableposition=top]{caption}
\usepackage{cite}

\ifCLASSINFOpdf
\else
\graphicspath{{E:/paper/Dr.Kalbasi/pictures/}}
\fi

\newtheorem{theorem}{Theorem}
\newtheorem{example}{Example}
\newtheorem{corollary}[theorem]{Corollary}
\newtheorem{lemma}[theorem]{Lemma}
\renewcommand{\footnoterule}{%
	\kern-3pt                         
	\hrule width 0.25\columnwidth height 0.4pt   
	\kern2.6pt                          
}
\begin{document}
	%
	\title{Random generation of group elements using\\ combinatorial group theory and automata theory, along with a hardware example\\}
	%
	%
	%
	
	\author{MohammadJavad~Vaez,
		Marjan~Kaedi,
		and~Mahdi~Kalbasi
		\thanks{M. Vaez is a master student in the School of Mathematics, Statistics, and Computer Science, College of Science, University of Tehran, Tehran, Iran.
		\newline
		e-mail: (\texttt{mohammadjavadvaez@gmail.com}).}
		\thanks{M. Kaedi (\texttt{kaedi@eng.ui.ac.ir}) and M. Kalbasi (\texttt{mahdikalbasi@gmail.com}) are professors in Faculty of Computer Engineering, University of Isfahan, Isfahan, Iran.}
	}

	\maketitle
	
	\begin{abstract}
		In this paper, we introduce a novel approach for generating random elements of a finite group given a set of generators of that. Our method draws upon combinatorial group theory and automata theory to achieve this objective. Furthermore, we explore the application of this method in generating random elements of a particularly significant group, namely the symmetric group (or group of permutations on a set). Through rigorous analysis, we demonstrate that our proposed method requires fewer average swaps to generate permutations compared to existing approaches. However, recognizing the need for practical applications, we propose a hardware-based implementation based on our theoretical approach, and provide a comprehensive comparison with previous methods. Our evaluation reveals that our method outperforms existing approaches in certain scenarios. Although our primary proposed method only aims to speed up the shuffling and does not decrease its time complexity, we also extend our method to improve the time complexity.
	\end{abstract}
	

	%
	\IEEEpeerreviewmaketitle

	\section{Introduction}
	%
	%
	%
	%
	\IEEEPARstart{A} permutation is a bijection from one set to itself. Roughly speaking, it is a rearrangement or shuffling of a set of elements. Generating random permutations has diverse applications across different branches of computer science, such as cybersecurity (including cryptography \cite{Andreeva2015-di}, image encryption \cite{Wang2021-vz}, biometric template security \cite{Punithavathi2021-sg}, secure machine learning \cite{Zheng2022-im}), randomized algorithms \cite{Kao1996-yo,Motwani1995-ww}, Monte Carlo simulation and randomization tests \cite{Hemerik2018-ue,Berry2014-hp,Li2013-la,Manly2018-uf}, machine learning \cite{Mishchenko2020-ei,Zheng2022-im}, and other miscellaneous algorithms \cite{Gan2010-ya}. The extensive range of these applications motivates the search for faster RPG methods.
	
	Arguably, the most well-known algorithm for this purpose is the Fisher-Yates algorithm. Ronald A. Fisher and Frank Yates introduced one of the first algorithms for random permutation generation (RPG) having $O(n^2)$ time complexity and $O(n)$ space complexity \cite{Fisher1938-aq}. Some decades later, an improved algorithm was introduced by Richard Durstenfeld, which had $O(n)$ time complexity and $O(1)$ space complexity \cite{Durstenfeld1964-yr}. This algorithm was popularized after being introduced in Knuth's \textit{The Art of Computer Programming}. Knuth attributes this algorithm to Fisher and Yates, and its computer implementation to Durstenfeld \cite{Knuth1998-mz}. However, according to \textit{A Historical Note on Shuffle Algorithms}, the Durstenfeld algorithm was a new RPG algorithm when introduced in 1964 \cite{OConnor2014-bq}. This historical point is the reason why Durstenfeld algorithm is sometimes called Fisher-Yates shuffle \cite{Arndt2010-nt} and sometimes Knuth shuffle \cite{Odom2019-co}. Here, we follow this misnomer and use the term \textquotedblleft Fisher-Yates algorithm\textquotedblright~to refer to the algorithm introduced by Durstenfeld! The pseudocode of this algorithm for a zero-based array $A$, is as follows \cite{OConnor2014-bq}:
	
	\begin{algorithm}[H]
		\caption{Descending Fisher-Yates Shuffling Algorithm}
		\begin{algorithmic}
			\For{$i \gets n-1$ \textbf{downto} $1$}
			\State Choose a random number $j$ from $\{0, 1, \ldots, i\}$
			\State Swap $A[i]$ and $A[j]$
			\EndFor
		\end{algorithmic}
	\end{algorithm}

	There is also an equivalent ascending version of this algorithm \cite{Arndt2010-nt}:
	
	\begin{algorithm}[H]
		\caption{Ascending Fisher-Yates Shuffling Algorithm}
		\begin{algorithmic}
			\For{$i \gets 0$ \textbf{to} $n-2$}
			\State Choose a random number $j$ from $\{i, i+1, \ldots, n-1\}$
			\State Swap $A[i]$ and $A[j]$
			\EndFor
		\end{algorithmic}
	\end{algorithm}

	It is usually important to consider shuffling the array $A=\{0,1,\dots,n-1\}$; since permutations of this set can be easily extended to any n-element array by a bijection \cite{Odom2019-co}. That is why Knuth also suggested a modification when we just want a random permutation of the integers $\{1,2,\dots,n\}$ in order to avoid swapping \cite{Knuth1998-mz}. Yet it still needs a for loop from $0$ to $n-2$. The hardware corresponding to this algorithm has been implemented and evaluated too \cite{Odom2019-co}, and the number of clock cycles it needs is a multiple of $n-1$ in different implementations. However, we will see later that the expected number of swaps required to generate a random permutation of n elements is $n-H_n$, and we present a new randomized method that generates permutations with this number of swaps.
	
	This paper begins by presenting proofs of combinatorial, algebraic, and probabilistic facts about permutation groups. Next, we introduce an accelerated hardware method for shuffling. Finally, we extend our method to enhance the time complexity of RPG.
	
	\hfill
	
	\section{Mathematical background}
	Due to the diverse insights covered in this paper, providing an exhaustive introduction to all the necessary background mathematics would be digressive. Hence, we present essential facts from combinatorial group theory and probabilistic automata, sourced from \cite{Magnus2004-in,Vidal2005-xv}. For readers who are unfamiliar with the basic concepts of group theory, particularly symmetric groups, and automata theory, we recommend referring to \cite{Malik1997-mj} and \cite{Linz2017-vb}, respectively.
	\subsection{Combinatorial Group Theory}
\begin{mdframed}[
	backgroundcolor=gray!20, 
	leftline=true, 
	linecolor=gray!20, 
	linewidth=4pt 
	]
	\begin{quote}
		\textbf{Definition.} Let $a, b, c, \ldots$ be distinct symbols and form the new symbols $a^{-1}, b^{-1}, c^{-1}, \ldots$. A word $W$ in the symbols $a, b, c, \ldots$ is a finite sequence $f_1, f_2, \ldots, f_{n-1}, f_n$, where each of the $f_\nu$ is one of the symbols $a, b, c, \ldots, a^{-1}, b^{-1}, c^{-1}, \ldots$.
		The length $L(W)$ of $W$ is the integer $n$. For convenience, we introduce the empty word of length zero and denote it by $1$. If we wish to exhibit the symbols involved in $W$, we write $W(a, b, c, \ldots)$.
		
		It is customary to write the sequence $f_1, f_2, \ldots, f_{n-1}, f_n$ without the commas$\ldots$.
		The inverse $W^{-1}$ of a word $W = f_1 f_2 \ldots f_{n-1} f_n$ is the word $f_n^{-1} f_{n-1}^{-1} \ldots f_2^{-1} f_1^{-1}$, where if $f_\nu$ is $a$ or $a^{-1}$, then $f_\nu^{-1}$ is $a^{-1}$ or $a$, respectively. Similarly, if $f_\nu$ is one of the symbols $b$ or $b^{-1}$, $c$ or $c^{-1}$, $\ldots$, the inverse is obtained by taking the inverse of the symbol. The inverse of the empty word is itself.
		
		$\ldots$
		
		If $W$ is the word $f_1 f_2 \ldots f_n$ and $U$ is the word $f_1' f_2' \ldots f_r'$, then we define their juxtaposed product $WU$ as the word $f_1 f_2 \ldots f_n f_1' f_2' \ldots f_r'$
		
		$\ldots$
		
		Given a mapping $\alpha$ of the symbols $a, b, c, \ldots$ into a group $G$ with $\alpha(a) = g, \alpha(b) = h, \alpha(c) = k, \ldots$, then we say that (under $\alpha$) $a$ defines $g$, $b$ defines $h$, $c$ defines $k$, $\ldots$, $a^{-1}$ defines $g^{-1}$, $b^{-1}$ defines $h^{-1}$, $c^{-1}$ defines $k^{-1}$, $\ldots$. Moreover, if $W = f_1 f_2 \ldots f_{n-1} f_n$, then $W$ defines the element, denoted $W(g, h, k, \ldots)$, in $G$ given by $g_1 g_2 \ldots g_{n-1} g_n$ where $f_\nu$ defines $g_\nu$; the empty word $1$ defines the identity element $1$ of $G$.
		
		Clearly, if the words $U$ and $V$ define the elements $p$ and $q$ of $G$, then $U^{-1}$ defines $p^{-1}$ and $UV$ defines $pq$ \cite{Magnus2004-in}.
	\end{quote}
\end{mdframed}
	Given a group $G$ and a set of words defining the elements of $G$, we can introduce an equivalence relation between words in this way:
	\begin{equation} \label{eq1}
		W_1 \sim W_2
	\end{equation}
	if they define the same element in $G$ \cite{Magnus2004-in}. For example, let $G$ be the symmetric group $S_3$ and let $\alpha$ be the mapping $a \mapsto (1,2)$, $b \mapsto (1,3)$, $c \mapsto (2,3)$. Then $ab \sim ca$ because both $ab$ and $ca$ define the permutation $(1,3,2)$.
\begin{mdframed}[
	backgroundcolor=gray!20, 
	leftline=true, 
	linecolor=gray!20, 
	linewidth=4pt 
	]
	\begin{quote}
		The class of all words in $a, b, c, \ldots$ equivalent to $W$ will be denoted by $\{W\}$, and $W$ or any other word contained in $\{W\}$ will be called a representative of $\{W\}$. We introduce multiplication of equivalence classes by:
		
		\begin{equation} \label{eq2}
			\{W_1\} \cdot \{W_2\} = \{W_1 W_2\}
		\end{equation}
		...
		
		\begin{theorem}
			The set $G$ of equivalence classes of words in $a, b, c, \ldots$ defined by the relation $\sim$ in \eqref{eq1} is a group under the multiplication defined by \eqref{eq2} \cite{Magnus2004-in}.
		\end{theorem}
		Let $a_1, a_2, \ldots, a_n$ be the generators of group $G$.
		Define an order relation $<$ among the words $W(a_1, a_2, \ldots, a_n)$ as follows:
		
		If $L(W_1) < L(W_2)$, then $W_1 < W_2$;
		
		$a_1 < a_1^{-1} < a_2 < a_2^{-1} < \ldots < a_n < a_n^{-1}$
		
		If $L(W_1) = L(W_2)$ and $W_1$ and $W_2$ first differ in their $k$-th terms, then order $W_1$ and $W_2$ according to their $k$-th terms. For example,
		$1 < a_1 < a_2 a_n < a_2 a_n^{-1} < a_1^3$ \cite{Magnus2004-in}.
	\end{quote}
\end{mdframed}

	If we select a unique representative from each equivalence class of words, we call that a canonical form. One method for presenting the group $G$ as a set of canonical forms is to choose the \textquotedblleft least\textquotedblright~element in each equivalence class \cite{Magnus2004-in}. In this paper, we call this set \textquotedblleft standard representative system.\textquotedblright
	
	So far, we have seen that a group can be represented as a set of words (strings). In order to randomly generate the elements of a group, we must assign the same probability to them (or equivalently, we must generate a uniform distribution over the group). A well-known tool to generate distributions over sets of (possible infinite cardinality) words is a probabilistic finite-state automaton (PFA) \cite{Vidal2005-xv}. Here, we will just have a cursory look at this tool and refer the interested readers to \cite{Vidal2005-xv} for further details.
	
	\subsection{Probabilistic Automata}
	The following part is taken from \cite{Vidal2005-xv}.
\begin{mdframed}[
	backgroundcolor=gray!20, 
	leftline=true, 
	linecolor=gray!20, 
	linewidth=4pt 
	]
	\begin{quote}
		\textbf{Definition.} A PFA is a tuple $\mathcal{A}=\langle Q_\mathcal{A},\Sigma,\delta_\mathcal{A},I_\mathcal{A},F_\mathcal{A},P_\mathcal{A}\rangle$ where:
		\begin{itemize}
			\item $Q_\mathcal{A}$ is a finite set of states;
			\item $\Sigma$ is the alphabet;
			\item $\delta_\mathcal{A}\subseteq Q_\mathcal{A}×\Sigma×Q_\mathcal{A}$  is a set of transitions;
			\item $I_\mathcal{A}:Q_\mathcal{A}\longrightarrow \mathbb{R}^{\geq0}$ (initial-state probabilities);
			\item $P_\mathcal{A}:\delta_\mathcal{A}\longrightarrow\mathbb{R}^{\geq0}$ (transition probabilities);
			\item $F_\mathcal{A}: Q_\mathcal{A}\longrightarrow\mathbb{R}^{\geq0}$ (final-state probabilities);
		\end{itemize}
		$I_\mathcal{A},P_\mathcal{A}$, and $F_\mathcal{A}$ are functions such that:
		\[\sum_{q\in Q_\mathcal{A}}^{}I_\mathcal{A}=1,\]
		and
		\[\forall q\in Q_\mathcal{A},F_\mathcal{A}(q)+\sum_{a\in \Sigma,q'\in Q_\mathcal{A}}^{}P_\mathcal{A}(q,a,q')=1\]
		$P_\mathcal{A}$ is assumed to be extended with $P_\mathcal{A}(q,a,q')=0$ for all $(q,a,q')\notin\delta_\mathcal{A}$.
		In what follows, the subscript $\mathcal{A}$ will be dropped when there is no ambiguity.
		
		$\ldots$
		
		\textbf{Definition.} A PFA $\mathcal{A}=\langle Q,\Sigma,\delta,I,F,P\rangle$ is a DPFA, if:
		\begin{itemize}
			\item $\exists q_0\in Q$ (initial state), such that $I(q_0)=1$;
			\item $\forall q\in Q,\forall a\in\Sigma,|\left\{q':(q,a,q')\in\delta\right\}|\leq1$.
		\end{itemize}
		In a DPFA, a transition $(q,a,q')$ is completely defined by $q$ and $a$ and a DPFA can be more simply denoted by $\langle Q,\Sigma,\delta,q_0,F,P\rangle$.
		
		$\ldots$
		
		PFA are stochastic machines that may not generate a probability space but a subprobability space over the set of finite-strings $\Sigma^*$. Given a PFA $\mathcal{A}$, the process of generating a string proceeds as follows:
		\begin{itemize}
			\item Initialization: Choose (with respect to a distribution $I$) one state $q_0$ in $Q$ as the initial state. Define $q_0$ as the current state.
			\item Generation: Let $q$ be the current state. Decide whether to stop, with probability $F(q)$, or to produce a move $(q,a,q')$ with probability $P(q,a,q')$, where $a\in\Sigma$ and $q'\in Q$. Output $a$ and set the current state to $q'$.
		\end{itemize}
	If PFA generates finite-length strings, a relevant question is that of computing the probability that a PFA $\mathcal{A}$ generates a string $x \in \Sigma^*$. To deal with this problem, let $\theta=(s_0,x_1',s_1,x_2',\ldots,s_{k-1},x_k',s_k)$ be a path for $x$ in $\mathcal{A}$; that is, there is a sequence of transitions $(s_0,x_1',s_1), (s_1,x_2',s_2), \ldots, (s_{k-1},x_k',s_k) \in \delta$ such that $x=x_1' x_2'\ldots x_k'$. The probability of generating such a path is:
	\[\Pr_\mathcal{A}(\theta)=I(s_0)\cdot\left(\prod_{j=1}^k P(s_{j-1},x_j',s_j)\right)\cdot F(s_k)\].
	
	\textbf{Definition.} A valid path in a PFA $\mathcal{A}$ is a path for some $x \in \Sigma^*$ with probability greater than zero. The set of valid paths in $\mathcal{A}$ will be denoted as $\Theta_\mathcal{A}$.
	
	\textbf{Definition.} A state of a PFA $\mathcal{A}$ is useful if it appears in at least one valid path of $\Theta_\mathcal{A}$.
	
	\textbf{Proposition.} A PFA is consistent if all its states are useful \cite{Vidal2005-xv}.
	\end{quote}
\end{mdframed}
	\textbf{Definition.} In a similar manner, a useful state in a deterministic finite automaton (DFA) refers to a state that is reachable from the initial state and can eventually lead to an accepting state. Conversely, a state that cannot fulfill these criteria is termed useless and can be eliminated from the DFA without impacting its functionality.

	It is worth mentioning that algebraic insights have been widely used when dealing with permutations and permutation puzzles \cite{Joyner2008-gs}. Furthermore, there is a strong connection between algebraic structures (especially semigroups and groups) and automata theory \cite{Holcombe2004-lb, Godin2017-rd}. In this paper, we combined these branches to introduce a new method. In fact, we have used 5 different insights interchangeably. They have been shown in table \ref{tab:insights}, and we will explain them in the sequel.
	
	\begin{table}[h]
		\centering
		\caption{Different insights used in this paper interchangeably}
		\label{tab:insights}
		\begin{tabular}{|l|l|l|l|l|}
			\hline
			\textbf{Algebraic insight} & \textbf{Language-theoretic insight} & \textbf{Automatic insight}$^*$ & \textbf{Machinelike insight}$^\dag$ & \textbf{Graphical insight} \\ 
			\hline
			group                     & language                            & DFA                       & finite-state machine (FSM)                          & directed graph            \\ 
			\hline
			group element             & word                                & input                     & final output                 & path starting from the initial state \\ 
			\hline
			-                         & -                                    & state                     & state                        & node                       \\ 
			\hline
			generator                 & symbol                              & transition                & action (output)              & edge                       \\ 
			\hline
		\end{tabular}
		\\[8pt]  
		{\footnotesize $^*$One of the meanings of \textquotedblleft automatic\textquotedblright~is automaton-like \cite{Harington1897-oh}. In the context of algebra, it can also mean related to an automaton \cite{Gradel2020-ur} or having one or more finite-state automata \cite{WikiDiff2018-hy}. $^\dag$ For more information about the machinelike insight and additional topics concerning modeling and implementation of FSMs, please refer to \cite{Lee2017-rl}. As we will see later, the automatic and machinelike insights are inverse of each other. But here, our purpose is to generate all permutations of group $S_n$ with the same probability. Since each group is closed with respect to inversion, the output of the FSM constructs the permutations too.}
	\end{table}
	
	\section{Proposed Method}
	The main idea of this paper is made up of 4 steps:
	\begin{enumerate}
		\item presenting the symmetric group $S_n$ as a language called $L_n$
		\item obtaining the minimal DFA of language $L_n$
		\item calculate the probability of each transition in order to generate all permutations equally likely
		\item designing a piece of hardware for shuffling
	\end{enumerate}
	We explain each step through an example.
	
	\begin{example}
		Consider the symmetric group $S_4$ containing all possible permutations on a 4-element set.
		
		\textbf{Step 1)} \label{step1} Here we have decomposed all permutations into transpositions (except the identity permutation which we do not need to factorize).
		
		Let $\alpha$ be the mapping $a \mapsto (1,2)$, $b \mapsto (1,3)$, $c \mapsto (2,3)$, $d \mapsto (1,4)$, $e \mapsto (2,4)$, $f \mapsto (3,4)$.\footnote{In some books and papers, cycles are written without comma.} Then the group $S_4$ under $\alpha$ will be presented as follows:
		\begin{align*}
			S_4 = \{ & 1, a, b, c, d, e, f, (1,2,3)=ba, (1,3,2)=ab, (1,2,4)=da, (1,4,2)=ad, (1,3,4)=db,(1,4,3)=bd,(2,3,4)=ec, \\
			& (2,4,3)=ce,(1,2)(3,4)=af,(1,3)(2,4)=be,(1,4)(2,3)=dc, (1,2,3,4)=dba,(1,2,4,3)=bda,\\
			& (1,3,2,4)=dab,(1,3,4,2)=adb,(1,4,2,3)=bad,(1,4,3,2)=abd \}.
		\end{align*}
		This is one of many possible presentations of group $S_4$. To obtain this for each disjoint cycle we used the fact that $(i_1,i_2,\ldots,i_k)=(i_1,i_k)(i_1,i_{k-1})\ldots(i_1,i_2)$. For example $(1,2,3,4)=(1,4)(1,3)(1,2)=dba$. However other presentations are accepted too.

		Here we have presented $S_4$ as if it is a language whose alphabet is the set of transpositions so that we can obtain an automaton for it.
		
		\textbf{Step 2)} The minimal DFA for such a language is depicted in figure \ref{fig:first_S4}.
		\begin{figure}[h]
			\centering
			\includegraphics[scale=0.7]{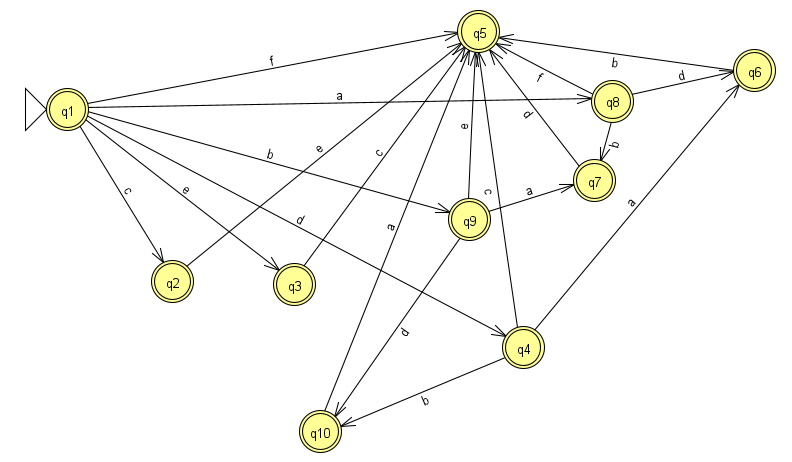}
			\caption{the DFA corresponding to the group $S_4$ (first presentation)}
			\label{fig:first_S4}
		\end{figure}

		\textbf{Step 3)} Now, we assign a probability to each transition. These probabilities must be calculated in such a way that all permutations are generated equally likely. Theorem \ref{prob_thm} will help us satisfy this condition.
		
		Note that from this stage onwards, we will use the opposite insight of step \ref{step1}. In the first step, there was an acceptor which would take a word as an input and move between states step by step. In each step, it would consume one symbol from the beginning of the word. Here, however, there is a machine that moves between states and applies a transposition to an array. Using language-theoretic insight, in each step, it produces one symbol and places it at the beginning of a word. So the ultimate output of this machine is a word. Hence, the set of words produced by the DPFA is equal to ${L_4}^{-1}$. Since $L_4$ presents $S_4$, ${L_4}^{-1}$ presents ${S_4}^{-1}$ which is equal to $S_4$~\footnote{For a set $A$, we define $A^{-1}=\left\{a^{-1}|a\in A\right\}$. Of course, the inversion is inherently different in the groups and among words.}.  For instance, $abd$ is a path in figure \ref{fig:first_S4}. Then $dba$ is the result of corresponding actions, since applying $a, b,$ and $d$ consecutively, constructs a composite function $d(b(a(1)))=dba$.
		
		\textbf{Step 4)} The last step is to map the DPFA to a piece of hardware.
		
		We will explain steps 3 and 4 further later.
		
		Although the DFA shown in figure \ref{fig:first_S4} is minimal, there could be fewer number of states using another presentation for group $S_4$. For example, the standard representative system of group $S_4$ under mapping $\alpha$ is as follows:
		\begin{align*}
			S_4 = \{ & 1, a, b, c, d, e, f, (1,2,3)=ac, (1,3,2)=ab, (1,2,4)=ae, (1,4,2)=ad, (1,3,4)=bf,(1,4,3)=bd,(2,3,4)=cf,\\
			& (2,4,3)=ce,(1,2)(3,4)=af,(1,3)(2,4)=be,(1,4)(2,3)=cd, (1,2,3,4)=acf,(1,2,4,3)=ace,\\
			& (1,3,2,4)=abe,(1,3,4,2)=abf,(1,4,2,3)=acd,(1,4,3,2)=abd \}.
		\end{align*}
		The minimal DFA for this presentation is shown in figure \ref{fig:second_S4}. As you can see, it has fewer states. It also has a more organized structure which we will discuss in the following theorem.
		
		\begin{figure}[h]
			\centering
			\includegraphics[scale=0.45]{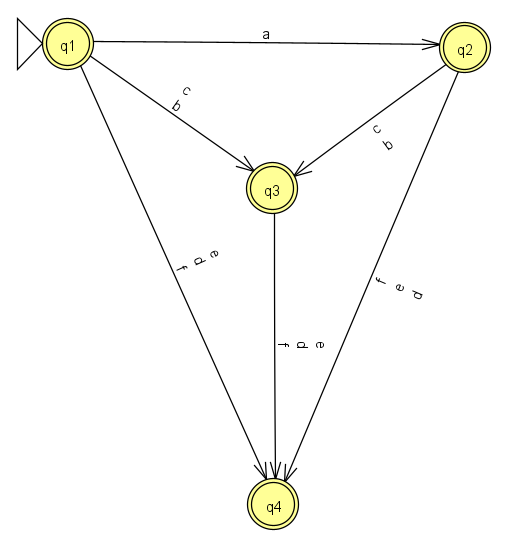}
			\caption{the DFA corresponding to the group $S_4$ (standard representative system)}
			\label{fig:second_S4}
		\end{figure}
	\end{example}
	
	\subsection{Theorems and Corollaries}
	\begin{theorem} \label{minimal_DFA}
		A minimal DFA of group $S_n$ is of the form $M_n=(Q,\Sigma,\delta,q_1,F)$;\footnote{It is more common to correspond the alphabet to symbols like $a,b,c,\ldots$ or $a_1,a_2,a_3,\ldots$. However, here we have used the transpositions for convenience.} \footnote{It is better to assume $n>1$ in order not to have an empty alphabet.} where
		\begin{align*}
			Q & = \{q_1, q_2, \ldots, q_n, q_{n+1}\} \\
			F & = \{q_1, q_2, \ldots, q_n\} = Q \setminus \{q_{n+1}\} \\
			\Sigma & = \{(i,j) | 1 \leq i < j \leq n\} \\
			\delta(q_i, (j,k)) & = 
			\begin{cases} 
				q_{\max\{j,k\}} & \text{if } k > i \\
				q_{n+1} & \text{if } k \leq i 
			\end{cases}
		\end{align*}
		
		In other words, it has the following properties:
		\begin{itemize}
			\item It has $n+1$ states, and all of them are final states except the last one, which is the trap state. We usually neglect the trap state and the transitions ending to that.
			\item If $i<j$, there are $j-1$ transitions from $q_i$ to $q_j$ corresponding to the transpositions $(x,j)$ where $1 \leq x \leq j-1$.
		\end{itemize}
		This DFA is unique up to isomorphism; i.e., we will have another minimal DFA by relabeling the numbers. However, for the sake of simplicity, we just work with this standard form and prove the following theorems based on that.
	\end{theorem}
	\begin{proof}
		The proof has three parts.
		\begin{itemize}
			\item The first part is to show that the language accepted by the DFA defined by this theorem, corresponds to the symmetric group.
			\item The second part is to show that there are no two different words defining the same permutation.\footnote{The second part is essential to prove that each permutation is generated just once.}
			\item The third part is to show that the DFA explained in the theorem is minimal.
		\end{itemize}
		
		Before proving the theorem, we give an example for $n=3$. If $n=3$, the minimal DFA is isomorphic to the DFA shown in figure \ref{fig:S3}, which has 3 states:
		\begin{figure}[h]
			\centering
			\includegraphics[scale=0.23]{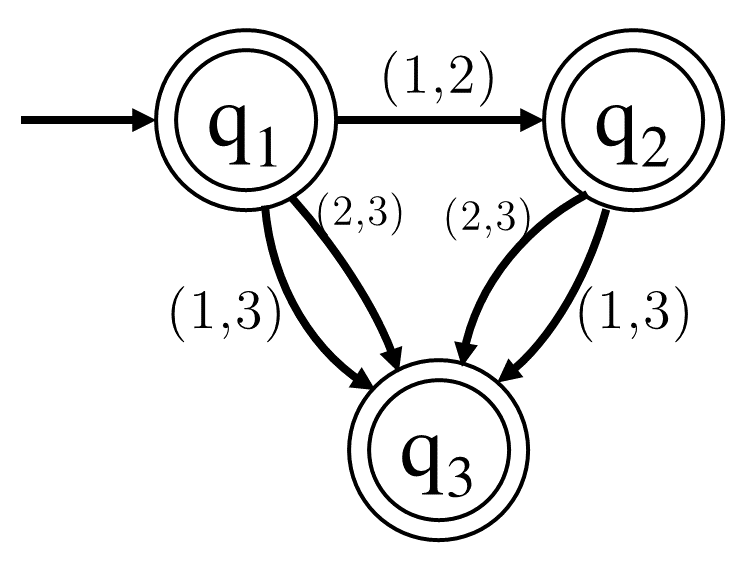}
			\caption{DFA corresponding to the group $S_3$ constructed based on theorem \ref{minimal_DFA}}
			\label{fig:S3}
		\end{figure}
		
		Then
		\begin{align*}
			\delta(q_1, (1,2)) & = q_2, \\
			\delta(q_1, (1,3)) & = \delta(q_1, (2,3)) = q_3, \\
			\delta(q_2, (1,3)) & = \delta(q_2, (2,3)) = q_3, \\
			\delta(q_i, (j,k)) & = \emptyset \quad \text{for } k \leq i
		\end{align*}
		So the group $S_3$ can be presented in this way: $S_3=\{1,a,b,c,ab,ac\}$ where $1$ is the identity permutation and $a$, $b$, and $c$ define transpositions $(1,2)$, $(1,3)$, and $(2,3)$ respectively.
		
		Now we prove the first and second parts of the theorem by induction. For $n=2$, the minimal DFA is shown in figure \ref{fig:S2}:
		\begin{figure}[h]
			\centering
			\includegraphics[scale=0.23]{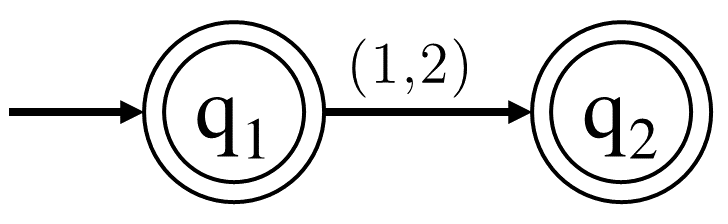}
			\caption{DFA corresponding to the group $S_2$ constructed based on theorem \ref{minimal_DFA}}
			\label{fig:S2}
		\end{figure}
		So the accepted words are $L_2=\left\{\epsilon,(1,2)\right\}$ which correspond to group $S_2$. Moreover, there are no two different words defining the same permutation.
		
		Now assume the proposition is true for $n=n_0$, i.e., the language accepted by the $M_{n_0}$ (which we call $L_{n_0}$) corresponds to the symmetric group $S_{n_0}$. In addition, there are no two different words defining the same permutation. Now we add a new node $q_{n_0+1}$ and connect every previous node to it through edges $(1,n_0+1),(2,n_0+1),\ldots,(n_0,n_0+1)$. For convenience, we consider its equivalent NFA (figure \ref{fig:NFA}).
		
		\begin{figure}[h]
			\centering
			\includegraphics[scale=0.3]{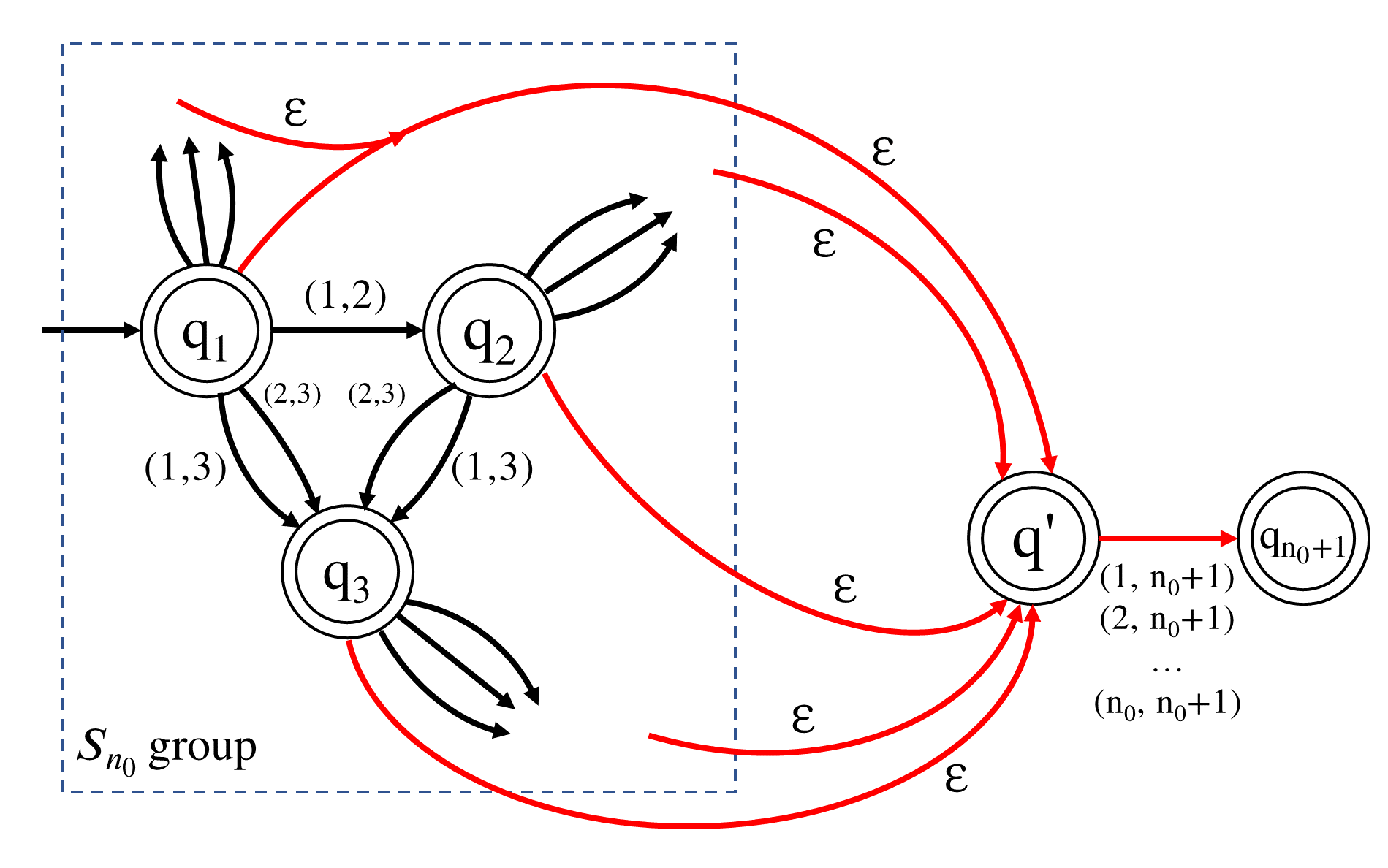}
			\caption{NFA equivalent to the group $S_{n_0+1}$ based on the construction explained in theorem \ref{minimal_DFA}}
			\label{fig:NFA}
		\end{figure}

		What we do is equivalent to connecting all previous states to a new state $q'$ by $\epsilon$-transitions and connecting $q'$ to $q_{n_0+1}$ through edges $(1,n_0+1),(2,n_0+1),\ldots, (n_0,n_0+1)$. The words accepted at the state $q'$ are the words accepted at states $q_1,q_2,\ldots,q_{n_0+1}$ which are equal to $S_{n_0}$ according to the induction hypothesis. Consider the language accepted by the whole NFA, which we call $L_{n_0+1}$. Our first goal is to show that $L_{n_0+1}=S_{n_0+1}$. First, note that
		\[S_{n_0+1}\supseteq L_{n_0+1}=S_{n_0}\cup S_{n_0}(1,n_0+1)\cup S_{n_0}(2,n_0+1)\ldots \cup S_{n_0}(n_0,n_0+1)\]
		where $Ab=\left\{ab|a\in A\right\}$ for a set $A$ and an element $b$ in group $S_{n_0+1}$.
		
		Furthermore, these sets are separate. Because
		\begin{itemize}
			\item Suppose there is a permutation $\pi \in S_{n_0}(i,n_0+1) \cap S_{n_0}(j,n_0+1)$. Then there exist permutations $\sigma_1, \sigma_2 \in S_{n_0}$ such that $\pi = \sigma_1(i,n_0+1) = \sigma_2(j,n_0+1)$. So $(j,n_0+1)(i,n_0+1) = \sigma_2^{-1} \sigma_1 \in S_{n_0}$, which is a contradiction.\footnote{Since $(j,n_0+1)(i,n_0+1)=(n_0+1,i,j)\notin S_{n_0}$}
			\item Now suppose there is a permutation $\pi \in S_{n_0} \cap S_{n_0}(i,n_0+1)$. Then there exist permutations $\sigma_1, \sigma_2 \in S_{n_0}$ such that $\pi = \sigma_1 = \sigma_2(i,n_0+1)$. So $(i,n_0+1) = \sigma_2^{-1} \sigma_1 \in S_{n_0}$ which is a contradiction.\footnote{Here we have used the properties of a group, including closure with respect to the group operation and invertibility of the elements.}
		\end{itemize}
		In addition, the cardinality of each set is $n!$ since the function $f:S_{n_0} \rightarrow S_{n_0}(i,n_0+1)$ such that $f(\pi)=\pi(i,n_0+1)$ is a bijection.
		
		As a result, the sets $S_{n_0}$, $S_{n_0}(1,n_0+1)$, $\ldots$, $S_{n_0}(n_0,n_0+1)$ partition the set $L_{n_0+1}$ as well as having the same cardinality. So
		\[|L_{n_0+1}| = |S_{n_0}| + |S_{n_0}(1,n_0+1)| + |S_{n_0}(2,n_0+1)| + \cdots + |S_{n_0}(n_0,n_0+1)| = (n_0+1)|S_{n_0}| = (n_0+1) \times n_0! = (n_0+1)!\]
		
		Now notice that based on the induction hypothesis, the words belonging to $L_{n_0}$ define distinct elements in group $S_{n_0}$. Hence for each $i$ such that $1 \leq i \leq n_0$ the words belonging to each $S_{n_0}(i,n_0+1)$ are distinct; because assuming $\pi_1(i,n_0+1) = \pi_2(i,n_0+1)$ for two permutations $\pi_1, \pi_2 \in S_{n_0}$ leads to $\pi_1 = \pi_2$. Using this result and the fact that the sets $S_{n_0}$, $S_{n_0}(1,n_0+1)$, $\ldots$, $S_{n_0}(n_0,n_0+1)$ partition the set $L_{n_0+1}$, we conclude that there are no repeating permutations in $L_{n_0+1}$.
		
		Since $L_{n_0+1} \subseteq S_{n_0+1}$ and they have the same finite cardinality, and there are no repeating permutations in $L_{n_0+1}$, we conclude that $L_{n_0+1} = S_{n_0+1}$.
		
		Now we prove that the DFA defined in the theorem is a minimal one. First, note that every state in the DFA $M_n$ is reachable; since for every $i\in\left\{2,\ldots,n\right\}$ the word $(1,i)$ puts the DFA in the state $q_i$. Furthermore, the initial state is reachable obviously.
		
		Now we prove that every two different states in the DFA are distinguishable, except the dead state, which we neglected in \cref{fig:NFA,fig:S2,fig:S3}. Therefore, we can partition the state set into final and nonfinal states to get the equivalence classes $\left\{q_1,q_2,\ldots,q_n\right\}$ and $\left\{q_{n+1}\right\}$.
		
		Now we split the first equivalency class, step by step. The state $q_1$ is distinguishable from other states, since $\delta(q_1,(1,2))=q_2$ which is final, but for every $i\in\left\{2,\ldots,n\right\}$, $\delta(q_i,(1,2))=q_{n+1}$ which is nonfinal. Likewise, $q_2$ is distinguishable from other states since $\delta(q_2,(1,3))=q_3$ which is final, but for every $i\in\left\{3,\ldots,n\right\}$, $\delta(q_i,(1,3))=q_{n+1}$ which is nonfinal. Moreover, we already proved that $q_1$ and $q_2$ are distinguishable. We can repeat this process for every state $q_k~(k<n)$. suppose we have proved that $q_1,q_2,\ldots,q_k$ are distinguishable. Also $q_k$ is distinguishable from next states; since $\delta(q_k,(1,k+1))=q_{k+1}$ which is final, but for every $i\in\left\{k+1,\ldots,n\right\}, \delta(q_i,(1,k+1))=q_{n+1}$ which is nonfinal. The last step is to prove that $q_n$ is distinguishable from others. However, this step has been proved through the previous steps.
	\end{proof}

	\begin{theorem} \label{remarks}
		Let $a_{1,1}, \ldots, a_{n-1,n}$ denote all the transpositions in which for all $i,j$ such that $1 \leq i < j \leq n$, $a_{i,j} \mapsto (i,j)$ is the mapping. We define an order relation $<_s$ among these symbols as follows:
		\[a_{h,i} <_s a_{j,k} \iff i <_s k \lor (i = k \land h <_s j)\]
		Let $<_w$ be the order relation among words induced by $<_s$.\footnote{That is $<_w$ extends $<_s$.} Suppose $L_n$ is the language accepted by $M_n$, the DFA defined in theorem \ref{minimal_DFA}. Then for each word $a \in L_n$, and another word $b$ such that $a \sim b$, we have $a <_w b$. In other words, the DFA $M_n$ defined in theorem \ref{minimal_DFA} accepts the canonical forms of group $S_n$ under mapping $a_{i,j} \mapsto (i,j)$ and relation $<_w$ ($\sim$ is the equivalence relation defined in \eqref{eq1}).
	\end{theorem}
	\begin{proof}
		In order to find the least element in each equivalency class, pay attention to the following remarks:
		\begin{enumerate}[a)]
			\item \label{remark_a} if a word $W$ has an equivalent word $V$ such that $V<W$, there will be no canonical form containing $W$ as a substring. Since for each two words $A$ and $B$, $V<W$ results in $AVB<AWB$.
			\item \label{remark_b}	Let $(h,i)$ and $(j,k)$ be two transpositions where $j<k$ and $h<i<k$. Then three cases may occur. \newline
			
			$\begin{cases} 
				\text{If } j=h,~~~~~~\text{then } (j,k)(h,i)=(h,k)(h,i)=(h,i,k) \sim (i,k,h)=(h,i)(i,k) \\
				\text{If } j=i,~~~~~~~\text{then } (j,k)(h,i)=(i,k)(h,i)=(i,k)(i,h)=(i,h,k) \sim (h,k,i)=(h,i)(h,k) \\
				\text{otherwise, } h, i, j, \text{and } k \text{ will be distinct.  so } (j,k)(h,i) \sim (h,i)(j,k)
			\end{cases}$
			\newline
			
			Hence, according to remark \ref{remark_a}, in each case $(j,k)(h,i)$ cannot be contained in a canonical form.
			\item \label{remark_c} Let $(i,k)$ and $(j,k)$ be two transpositions where $i<j<k$. Since $(i,k)(j,k)=(k,i)(k,j)=(k,j,i)\sim(i,k,j)=(i,j)(i,k)$ there will be no canonical forms containing $(i,k)(j,k)$ where $k>i,j$ according to remark \ref{remark_a}.
			\item \label{remark_d} Based on remarks \ref{remark_b} and \ref{remark_c}, we conclude that if $j<k$ and $h<i$, then $(j,k)(h,i)$ can be contained in a canonical form only if $k<i$.
			\item \label{remark_e} Let $(i_1,j_1)(i_2,j_2)\ldots(i_t,j_t)$ be a canonical form in which for each $k$, $i_k<j_k$. Then based on remark \ref{remark_d}, we have $j_1<j_2<...<j_t$. Note that the words having this form are exactly what the DFA accepts. Now we prove that all words having this form are canonical forms. For this purpose, we can arrange the transpositions as follows:
			\[
			\begin{aligned}
				(1,2) \\
				(1,3) & \quad (2,3) \\
				(1,4) & \quad (2,4) \quad (3,4) \\
				\ldots \\
				(1,n) & \quad (2,n) \quad \ldots \quad (n-1,n)
			\end{aligned}
			\]
		\end{enumerate}
	The words accepted by the DFA are constructed by selecting transpositions $a_{i_1,j_1},a_{i_2,j_2},\ldots,a_{i_t,j_t}$, such that $j_1<j_2<\ldots<j_t$. Of course, you can select no transpositions from some rows. Even you can select no transpositions at all, which results in the identity permutation.
	
	Now note that each sequence of transpositions out of the words accepted by the DFA has one of the following properties:
	\begin{enumerate}
		\item Including two transpositions from one row.
		\item Including transpositions $a_{i_1,j_1},a_{i_2,j_2}$ such that $j_1<j_2$ and $a_{j_2}$ comes before $a_{j_1}$ in the sequence.
	\end{enumerate}
	So they cannot be canonical forms according to remarks \ref{remark_b} and \ref{remark_a}, respectively. Given that all permutations are presented once in $L_n$, the words accepted by the DFA are the canonical forms.
	\end{proof}

	\begin{corollary} \label{minimum_length}
		Each word belonging to $L_n$ has the minimum length in its equivalency class.
	\end{corollary}
	\begin{theorem} \label{expected}
		The expected minimum number of transpositions in the decomposition of a permutation $\sigma\in S_n$ is $n-H_n$ where $H_n=\sum_{i=1}^{n}\frac{1}{i}$ is the $n'$th harmonic number \cite{Fialkow1992-sf}.
	\end{theorem}
	\begin{corollary} \label{avg_length}
		Using corollary \ref{minimum_length} and theorem \ref{expected}, the average length of permutations constructed by the DPFA is $n-H_n$.
	\end{corollary}

	In the next step, we must generate all the permutations with the same probability. To pursue this goal, we must assign a suitable probability to each transition to convert the DFA to a DPFA. The following theorem explains how to do this. Before going to the next theorem, note that if a language is finite, for every useful state $q_a$ in its DFA, if there is a transition from $q_a$ to $q_b$, there must not be any transitions from $q_b$ to $q_a$; otherwise it will lead to infinite number of words. Also note that this condition is weaker than being a directed acyclic graph (DAG).
	\begin{theorem} \label{prob_thm}
		Consider a DFA of a finite language $L$, starting from state $q_1$. Let $\pi_a$ be the number of paths starting from state $q_a$ (including paths of length zero) that end to a final state. Then, if we consider the following conditions for useful states, each word is generated with probability $\frac{1}{|L|}$:
		\begin{itemize}
			\item $I(q_1)=1$ and $I(q_a)=0$ for all $a\neq1$ (i.e. we always start from state $q_1$)
			\item $P(q_a,e,q_b)=\frac{\pi_b}{\pi_a}$ as the probability of transition from state $q_a$ to state $q_b$ through symbol $e$ (and $P(q_a,e,q_b)=0$ if $(q_a,e,q_b)\notin\delta$)
			\item $F(q_a)=\frac{\chi_F(q_a)}{\pi_a}$ as the probability of halting the generation process in state $q_a$ (where $\chi_F$ is an indicator function and returns $1$ if $q_a$ is a final state and $0$ if it is non-final)
		\end{itemize}
	\end{theorem}
	\begin{proof}
		First, we must show that the probabilities claimed in the theorem are well defined. It is obvious that all the defined probabilities are non-negative. Furthermore, $\sum_{q\in Q}^{}I(q)=1$. So we must check the second condition:
		\[\forall q_a\in Q,F(q_a)+\sum_{e\in\Sigma,q_b\in Q}^{}P(q_a,e,q_b)=1\]
		
		For convenience, we define the function $E:Q \times Q \rightarrow \mathbb{Z}^{\geq 0}$ with the function rule $E(q_a,q_b) =$ the number of edges connecting $q_a$ to $q_b$. Note that for every $b$ such that $b \neq a$, the number of paths starting with $q_a \rightarrow q_b$ is $\pi_b E(q_a,q_b)$. Therefore
		\[\pi_a = \chi_F (q_a) + \sum_{b \neq a} \pi_b E(q_a,q_b)\]
		Therefore
		\[\frac{\chi_F (q_a)}{\pi_a} + \sum_{b \neq a} \frac{\pi_b}{\pi_a} E(q_a,q_b) = 1\]
		In other words,
		\[F(q_a ) + \sum_{b \neq a} P(q_a,e,q_b )E(q_a,q_b ) = 1\]
		This can be written like this
		\[F(q_a ) + \sum_{i \text{ s.t. } b_i \neq a} \sum_{e \text{ s.t. } (q_a,e,q_{b_i}) \in \delta} P(q_a,e,q_{b_i}) = 1\]
		Note that for all $e$ such that $(q_a,e,q_b) \notin \delta$, $P(q_a,e,q_b) = 0$. Furthermore, since we have considered useful states, and based on the remark before the theorem, a useful state cannot have self-loop; otherwise it would create an infinite number of words. Therefore $P(q_a,e,q_a) = 0$
		Therefore
		\[F(q_a) + \sum_{b_i} \sum_{e \text{ s.t. } (q_a,e,q_{b_i}) \in \delta} P(q_a,e,q_{b_i}) = 1\]
		That is
		\[F(q_a) + \sum_{e \in \Sigma, q_b \in Q} P(q_a,e,q_b) = 1\]

		Now we are ready to prove that with these transition probabilities ($P(q_a,e,q_b)=\frac{\pi_b}{\pi_a}$), every word is generated with probability $\frac{1}{|L|}$.
		
		First, note that since every path corresponds to a specific word, the number of paths starting from node $q_1$ equals the number of language elements. Hence $\pi_1=|L|$. Now consider a specific word. In the DFA, it has such a form: \[q_1 \xrightarrow{e_1} s_2 \xrightarrow{e_2} s_3 \xrightarrow{e_3} \ldots \xrightarrow{e_{t-1}} s_t.\]
		So its production probability is $\frac{\pi_{s_2}}{\pi_1} \times \frac{\pi_{s_3}}{\pi_{s_2}} \times \ldots \times \frac{\pi_{s_t}}{\pi_{s_{t-1}}} \times\frac{1}{\pi_{s_t}} =\frac{1}{\pi_1} =\frac{1}{|L|}$.
	\end{proof}

	\begin{theorem} \label{number_of_paths}
		Consider the minimal DFA of group $S_n$. Let $a$ be an integer such that $1\leq a\leq n$ and $\pi_a$ be the number of paths starting from state $q_a$ (including paths of length zero). Then $\pi_a=\frac{n!}{a!}$.
	\end{theorem}
	\begin{proof}
		We prove the theorem by induction on the $a^{th}$ node.\footnote{In fact, the principle of induction is explained like this:
			
			Let $\alpha$ be an integer, and let $P(k)$ be a proposition about $k$ for each integer $k\geq\alpha$. Then if $P(\alpha)$ is true and $\forall k\geq\alpha: (P(k)\Longrightarrow P(k+1))$, we conclude that $P(k)$ is true for all integers $k\geq\alpha$.\newline
			However, here we have statement $P(k)$, which we want to prove for $k\leq n$. For this purpose, we can consider the statement $Q(k)=((k\leq n)\Longrightarrow P(k))$ and use the principle of induction for $Q$.} Obviously, the proposition is true for $a=1$; since the number of paths starting from state $q_1$ (which is the initial state) equals the total number of permutations which is $n!$.
		
		Now suppose that $\pi_{a_0} = \frac{n!}{a_0!}$ for $a_0 < n$. We want to show that $\pi_{a_0+1} = \frac{n!}{(a_0+1)!}$. Let $\mathscr{P}_i$'s be the sets of all paths starting from node $q_i$ $(1 \leq i \leq n)$. Consider the mapping $f: \mathscr{P}_{a_0} \rightarrow \mathscr{P}_{a_0+1}$. Consider an arbitrary path $e_1 e_2...e_t \in \mathscr{P}_{a_0+1}$ shown in figure \ref{fig:paths}. Since there are $a_0$ edges from node $q_{a_0}$ to node $q_{a_0+1}$, there are $a_0$ paths $\sigma e_1 e_2...e_t \in \mathscr{P}_{a_0}$ for different choices of $\sigma$. Moreover, $e_1 e_2...e_t \in \mathscr{P}_{a_0}$ because if $\delta(q_{a_0+1},e_1) = q_{k_1}$ then $\delta(q_{a_0},e_1) = q_{k_1}$.\footnote{Suppose $e_1=(i,j)$. Then $\delta(q_{a_0+1},e_1)=q_{k_1}\Longrightarrow j>a_0+1>a_0\Longrightarrow \delta(q_{a_0},e_1)=q_{k_1}$}
		
		As a result, for each arbitrary path starting from node $q_{a_0+1}$, there are exactly $(a_0+1)$ corresponding paths starting from node $q_{a_0}$. That is the function $f$ is a $(a_0+1)$-to-one correspondence. Therefore $\pi_{a_0} = (a_0+1) \pi_{a_0+1}$. Using the induction hypothesis, we obtain $\pi_{a_0+1} = \frac{\pi_{a_0}}{a_0+1} = \frac{\frac{n!}{a_0!}}{a_0+1} = \frac{n!}{(a_0+1)!}$.
	\end{proof}
	\begin{figure}[h]
		\centering
		\includegraphics[scale=0.3]{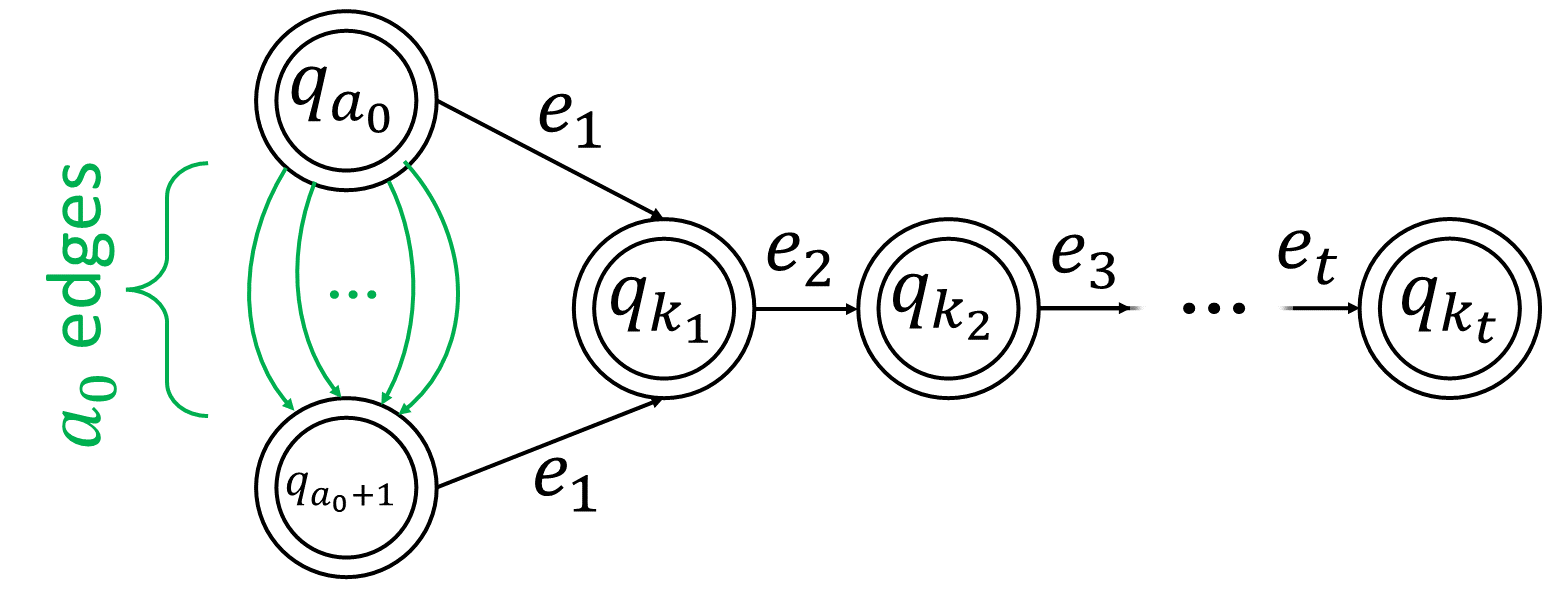}
		\caption{This picture aims to say for each arbitrary path starting from node $q_{a_0+1}$, there are exactly $a_0+1$ corresponding paths starting from node $q_{a_0}$}
		\label{fig:paths}
	\end{figure}

	\begin{corollary} \label{final_probs}
		from theorems \ref{prob_thm} and \ref{number_of_paths}, we conclude that if we consider $P(q_a,e,q_b)=\frac{\frac{n!}{b!}}{\frac{n!}{a!}}=\frac{a!}{b!}$ as the probability of transition from state $q_a$ to state $q_b$ through edge $e$, and F$(q_a)=\frac{1}{\frac{n!}{a!}}=\frac{a!}{n!}$ as the probability of halting the shuffling procedure in state $q_a$, every permutation is generated with probability $\frac{1}{n!}$.
	\end{corollary}

	\section{Hardware Design} \label{Hardware_Design}
	According to corollary \ref{final_probs}, we calculated the transition and termination probabilities. Now we are ready to design a hardware device that simulates the states and moves between them with corresponding transition probabilities or sends a terminate signal with corresponding final probabilities in order to make us understand that the permutation is ready.

	For instance, consider the DPFA corresponding to group $S_4$ and its transition table. It is shown in figure \ref{fig:array}. The table is filled in based on corollary \ref{final_probs}. In each node, $q(p)$ means that the process halts in $q$ with probability $p$. Furthermore, on each edge, the label $o(p)$ means that the DPFA will create the output $o$ with probability $p$ \cite{Vidal2005-xv}.
	
	\begin{figure}[h]
		\centering
		\includegraphics[scale=0.23]{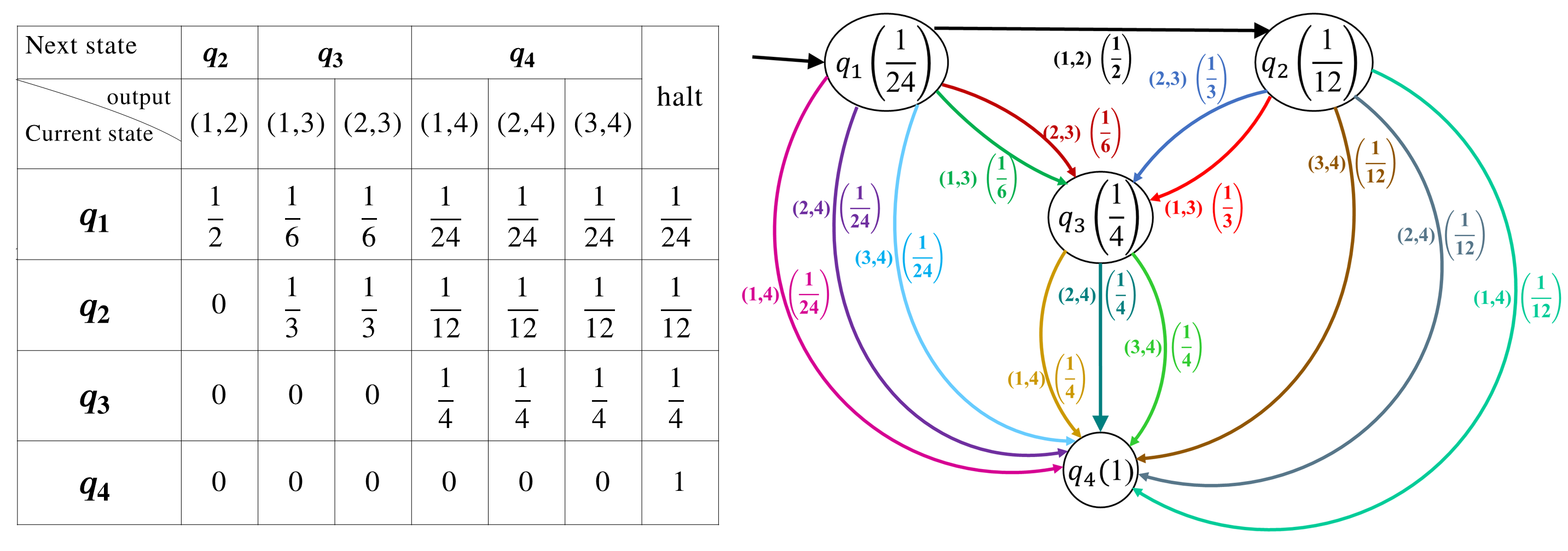}
		\caption{This picture illustrates the DPFA corresponding to group $S_4$, including both tabular and graphical representations.}
		\label{fig:array}
	\end{figure}
	
	Now we map the transition table of the DPFA to a ROM. Each address $a$ corresponds to state $q_a$. And the columns correspond to the transitions with the order explained in theorem \ref{expected}. We have the probabilities of each transition, so we use the idea of roulette wheel selection \cite{LIPOWSKI20122193}. Therefore, we can place the cumulative distribution function (CDF) of transitions at each row. However, instead, we multiply all the values by $n!$ in order to avoid struggling with floating-point numbers. Figure \ref{fig:first_hardware} illustrates the mapping of probabilities to the hardware for symmetric group $S_4$. Here we neglected the last column, which would contain $n!$ in each row.
	
	\begin{figure}[h]
		\centering
		\includegraphics[scale=0.25]{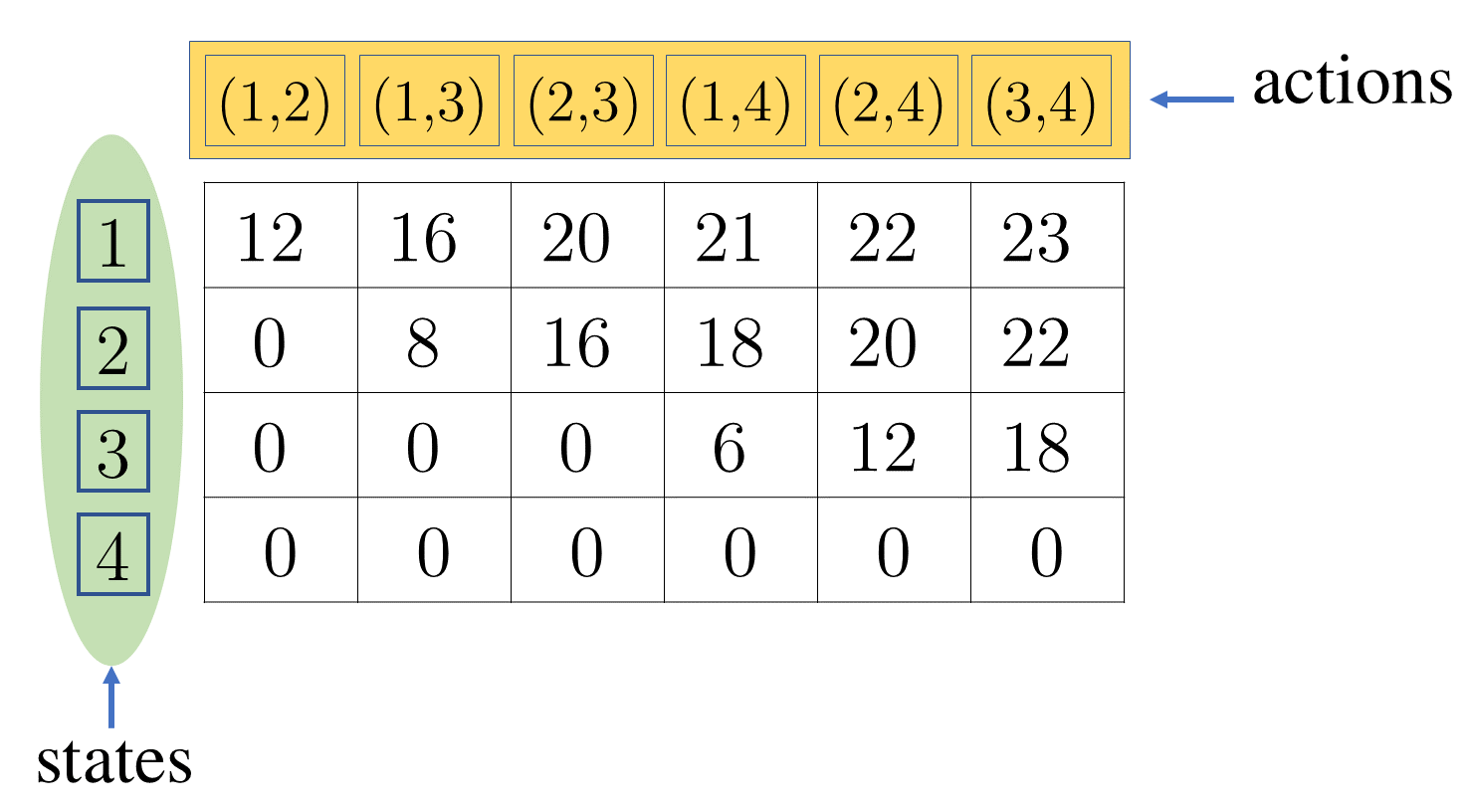}
		\caption{This piece of hardware consists of cumulative probabilities of transitions multiplied by $n!$ (Here $n=4$).}
		\label{fig:first_hardware}
	\end{figure}
	
	Since we want to consume fewer bits, we map each number to its previous number, which is a kind of relabeling (figure \ref{fig:second_hardware}).\footnote{Then the element at the address a under transposition $(i,j)$ will be $n!(a+1)!\left(\sum_{k=a+2}^{j}\frac{k-1}{k!}+\frac{i+1}{(j+1)!}[a<j]\right)$ in which $[\cdot]$ is the Iverson bracket notation \cite{Graham1989-tb}. However, it is faster to compute the numbers using dynamic programming.} However, since the highlighted row and column are actually virtual, nothing has changed so far. The difference will be in the circuit design.
	
	\begin{figure}[h]
		\centering
		\includegraphics[scale=0.5]{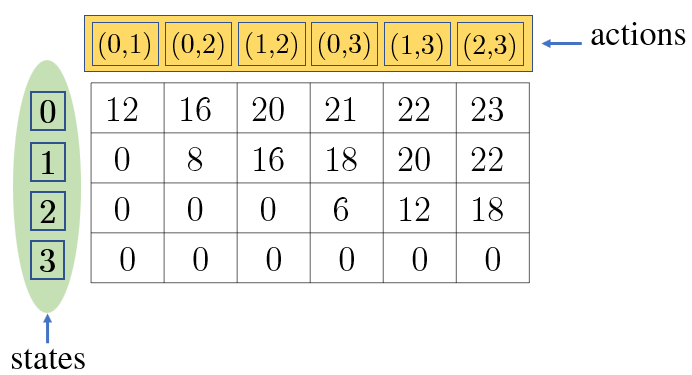}
		\caption{The label of addresses and transpositions of figure \ref{fig:first_hardware} has decreased by one, in order to consume fewer bits. This addressing is more common in hardware contexts.}
		\label{fig:second_hardware}
	\end{figure}
	
	Figure \ref{fig:undetailed_hardware} shows an abstract view of the complete hardware. Each time we want to generate a permutation, the state is set to $0$.\footnote{In fact, the input of the decoder must have a mux. However, as mentioned before, figure \ref{fig:undetailed_hardware} provides a high-level undetailed scheme.}  In each state for each column, the comparator outputs \textquotedblleft true\textquotedblright~(logical high) if the random number is greater than the corresponding number of that state and column. The gate \textquotedblleft index encoder\textquotedblright~is designed in such a way that generates the indices $i$ and $j$ corresponding to the column $(i,j)$. These indices will be passed to a true dual-port RAM (true DPRAM) in order to swap the contents of addresses $i$ and $j$. Moreover, the next state is equal to $j$.
	
	Let $s$ be the current state, and a random number $r\in\left\{1,\ldots,n!\right\}$ be generated. After comparing $r$ with the numbers at the address $s$, the first column in which the result of comparing is \textquotedblleft false\textquotedblright~determines $(i,j)$. In other words, if we denote $s$ as an array, the first $k$ such that $r>s[k]$ determines $(i,j)$. If for each $k$, $r>s[k]$, the process terminates and the permutation will be ready. Since the numbers in each row are nondecreasing, it suffices to check whether $r>s[\text{last}]$, that is, the output of the last comparator determines whether to terminate the process or not.
	
	\begin{figure}[h]
		\centering
		\includegraphics[scale=0.23]{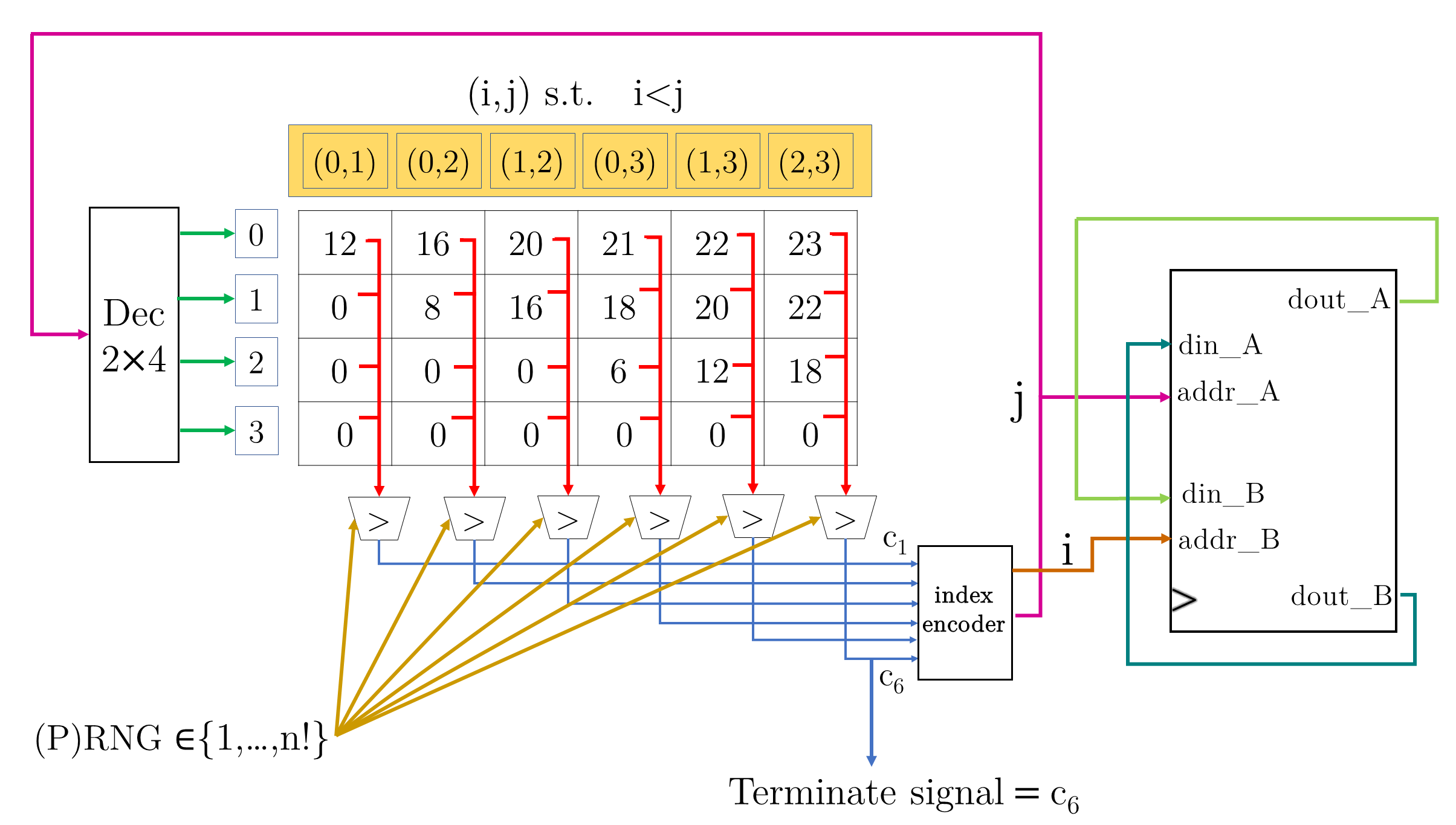}
		\caption{an undetailed view of a hardware device that makes a random permutation on a 4-element array}
		\label{fig:undetailed_hardware}
	\end{figure}
	
	\section{Performance and Complexity}
	When comparing two software or hardware algorithms, there can be used different aspects. For instance, for hardware implementations, space complexity, power, delay, PDP (power-delay product), area, fault tolerance and cost may matter. Here, we discuss the speed and complexity of the proposed method compared with the Fisher-Yates shuffle.
	
	\subsection{Comparing Performance with the Fisher-Yates Method} \label{comparing_performance}
	We can see, compared with Fisher-Yates hardware implementation, how much this hardware can decrease the expected time required to shuffle an $n$-element array for every specified $n$. For this purpose, first, we compare the expected number of required rounds each piece of hardware runs. Assuming $E_1$ and $E_2$ be the expected number of rounds needed in the Fisher-Yates and the proposed hardware, respectively\footnote{Of course, every implementation of the Fisher-Yates algorithm needs of $n-1$ rounds regardless of the resultant permutation. Hence, it needs $n-1$ rounds on average.}, we have:
	
	decrease percentage in the expected number of required rounds 
	$=\frac{|E_2-E_1|}{E_1}×100=\frac{|(n-H_n)-(n-1)|}{n-1}×100=\frac{H_n-1}{n-1}×100$
	
	Figure \ref{fig:rounds} provides a graph of percentage decrease in the expected required rounds versus the number of elements we want to permute. It shows that when $n\leq80$, using the proposed algorithm helps decrease the shuffling rounds, at least 5\%.\footnote{Assuming both pieces of hardware have the same clock frequency.}
	
	\begin{figure}[h]
		\centering
		\includegraphics[scale=0.7]{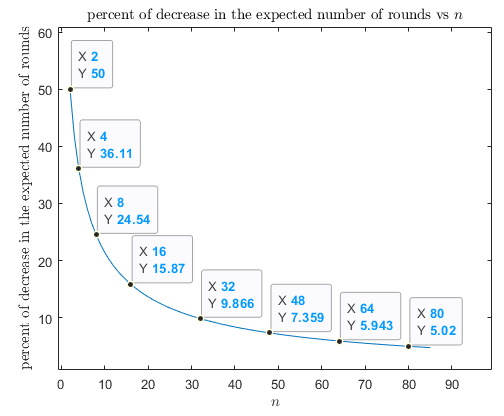}
		\caption{decrease percentage of the expected number of rounds vs. $n$ (compared with the Fisher-Yates hardware implementation)}
		\label{fig:rounds}
	\end{figure}
	
	Using another analysis, we can calculate the speed-up percentage. First, note that there are three different factors that affect the required shuffling time. The most high-level one is the number of rounds, which we discussed. The second one is the number of clock cycles each round has, and the third factor is the delay that logical gates have, which restricts the maximum possible clock frequency. Here, we do not consider the last factor because we have a high-level insight. Furthermore, it will be more significant for larger $n$'s, that is, when the circuits get larger and more complex. However, the advantage of the proposed hardware over the Fisher-Yates hardware vanishes as $n$ grows. Therefore, we do not apply the asymptotic analysis for this hardware.
	
	As a result, the most important factor after the number of rounds, is the number of clock cycles each round has. Compared with the Fisher-Yates hardware, our proposed hardware has fewer clock cycles in each round since its critical path is shorter. Because the existence of memory and swap are the same in both methods, except that in the proposed method, the memory is larger. But the Fisher-Yates hardware contains a counter as well \cite{Odom2019-co}, which makes the critical path longer. Nevertheless, since we are not going to discuss the implementation of hardware devices in this paper, we do not take this advantage into account.
	Therefore, assuming the clock frequency is the same in both implementations, we have: speed~$\propto\frac{1}{\text{time}}$; that is, if the time required to do a task multiplies by $k$, the speed of doing that task will multiply by $\frac{1}{k}$. Therefore, we have:
	
	speed-up percentage = $\frac{\text{second speed}-\text{first speed}}{\text{first speed}}×100=\left(\frac{\text{second speed}}{\text{first speed}}-1\right)×100=\left(\frac{1}{k}-1\right)×100$ 
	
	Assuming $k=\frac{n-H_n}{n-1}$ we conclude that: speed-up percentage = $\frac{H_n-1}{n-H_n}×100$. Figure \ref{fig:speed-up} provides a graph of speed-up percentage versus the number of elements we want to permute.
	
	\begin{figure}[h]
		\centering
		\includegraphics[scale=0.7]{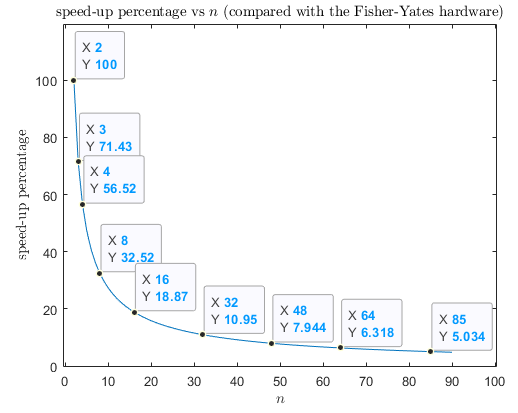}
		\caption{speed-up percentage vs. $n$ (compared with the Fisher-Yates hardware implementation)}
		\label{fig:speed-up}
	\end{figure}

	It is worth mentioning why the proposed method outperforms Fisher-Yates method. Consider the triangular scheme written below, in which $(~)$ means the identity permutation.
	
	\[
		\begin{aligned}
			(1,2) & \quad (~) \\
			(1,3) & \quad (2,3) \quad (~) \\
			(1,4) & \quad (2,4) \quad (3,4) \quad (~) \\
			\ldots \\
			(1,n) & \quad (2,n) \quad \ldots \quad (n-1,n) \quad (~)
		\end{aligned}
	\]
	Remind the remark \ref{remark_e} in theorem \ref{remarks}. As we explained, all the words accepted by the DFA are constructed by a top-down selection of exactly one element from each row. For example, $(1,2,3)=(1,2)(2,3)\overbrace{(\quad)\ldots(\quad)}^{n-3~ \text{times}}$. As a result, all the words produced by the DPFA are constructed by a bottom-up selection of exactly one element from each row. For instance, $(1,2,3)=\overbrace{(\quad)\ldots(\quad)}^{n-3~ \text{times}}(1,3)(1,2)$. This process is similar to the descending version of Fisher-Yates algorithm. For example, in case $n=4$, the complete state-space of the Fisher-Yates algorithm over time has depicted in figure \ref{fig:state-space}. At first, there is a 4-element array representing the identity permutation. In the $i$'th level $(i\geq1)$, the $(n+1-i)$th element of array obtained from the previous level will be swapped with an arbitrary element of its left side, or it remains at its previous position (Here $n=4$). Production of all permutations needs exactly $n-1$ levels. For example, the transposition $(2,3)$, which represents the array $1324$, is the result of selecting $(\quad),(2,3),$ and $(\quad)$ consecutively. However, in the proposed method, identity permutations do not waste a single level, and the expected number of levels needed to produce words will decrease. In this example, using the proposed method, the transposition $(2,3)$ will be generated in just one level; then, the procedure terminates.
	
	\begin{figure}[h]
		\centering
		\includegraphics[scale=0.23]{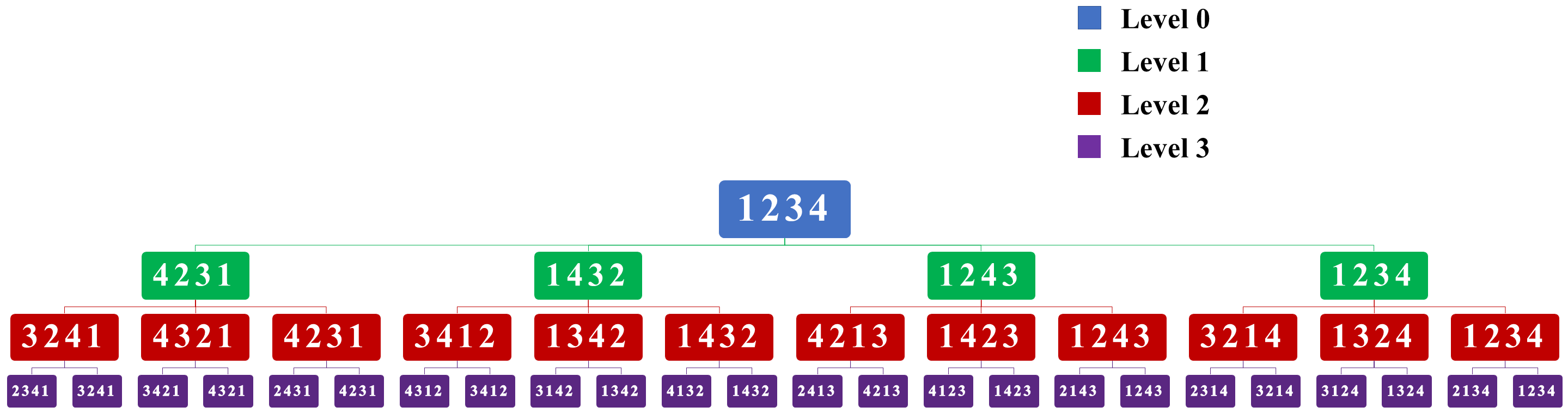}
		\caption{state-space of possible outcomes of the Fisher-Yates algorithm over time for $n=4$}
		\label{fig:state-space}
	\end{figure}
	
	\subsection{Time and Space Complexity} \label{time_and_space}
	According to corollary \ref{avg_length}, the number of random number generations and swaps to shuffle an $n$-element array in the proposed hardware is $n-H_n$. These operations are considered primitive operations, i.e., they can be done in $O(1)$ seconds. Therefore, the time complexity of the proposed method is $O(n)$. Furthermore, the space complexity of the proposed method is $O(n^4\log n)$ since the ROM has $n$ rows and $\frac{n(n-1)}{2}$ columns, and each column has the length $\lceil\log_2(n!)\rceil\in O(n\log n)$.
	
	We know that we can shuffle an array with $O(1)$ time complexity using a lookup table. That is, by storing all permutations in a ROM and generating a random number in $\left\{1,\ldots,n!\right\}$ we would access every permutation in $O(1)$ seconds. However, this method has the space complexity O$(n!×n\log n)$, which makes it impractical.
	
	Also, there are memoryless approaches that can generate a permutation having $O(1)$ time complexity, and generating a random permutation can be performed in just one clock cycle, albeit at a relatively low clock frequency. It is worth mentioning that since the nature of these designs needs similar or identical logic to be implemented a large number of times, these approaches will have a high area and delay growth as the number of inputs increases \cite{Odom2019-co}.
	
	Here, however, we can introduce a suite of hardware methods of the proposed approach to obtain different pieces of hardware and complexities. The idea is to use an arbitrary set of generators instead of transpositions. Let $H$ be a set of generators of group $G$, and $|H|=\gamma$. Remember the four steps we used to design the hardware with the transpositions as generating set. All the process will be the same for the set $H$, except we may not have the precalculated probabilities for transitions. Then we can find the transition probabilities using theorem \ref{prob_thm}. The larger the generating set is, the less the expected length of permutations will be.\footnote{Let $G_1$ and $G_2$ be the generating sets of group $S_n$ and $G_1\subsetneq G_2$. Then there exists a permutation $\sigma\in G_2\setminus G_1$. Therefore, the minimum length of presenting $\sigma$ will be shorter using $G_2$.}  Another point we must consider is to design hardware for each permutation in $H$ in order to perform them in $O(1)$ seconds.
	
	We can estimate a lower bound for the maximum length required to present all permutations using the generating set $H$. We call this number $l_{max}$. In the best case, all words from $0$ length, to the length $l_{max}$ define different elements of $G$. Hence
	\begin{equation}
		\overbrace{1+\gamma+\gamma^2+\ldots+\gamma^{l_{max}}}^{\text{number of words whose length}~\leq~l_{max}}\geq|G|
	\end{equation}
	That is
	\begin{equation} \label{identitiy}
		\frac{\gamma^{l_{max}+1}-1}{\gamma-1}\geq|G|
	\end{equation}
	This inequality helps us estimate a lower bound for $\gamma$ if we want to decrease the time complexity. For instance, if we want to lower the length of permutations of an $n$-element array to less than or equal to $\sqrt{n}\log_b(n)$, (where $b>1$), we can estimate a minimum $\gamma$, that is, the minimum cardinality the generating set must have. For this purpose, we can find the least $\gamma$ satisfying the condition $\frac{\gamma^{\lceil\sqrt{n}\log_b(n)\rceil+1}-1}{\gamma-1}\geq n!$. Figure \ref{fig:gamma_lower_bound} illustrates the lower bound of $\gamma$ for $n\leq20$ and $b=e$ (Euler's number).
		
	\begin{figure}[h]
		\centering
		\includegraphics[scale=0.5]{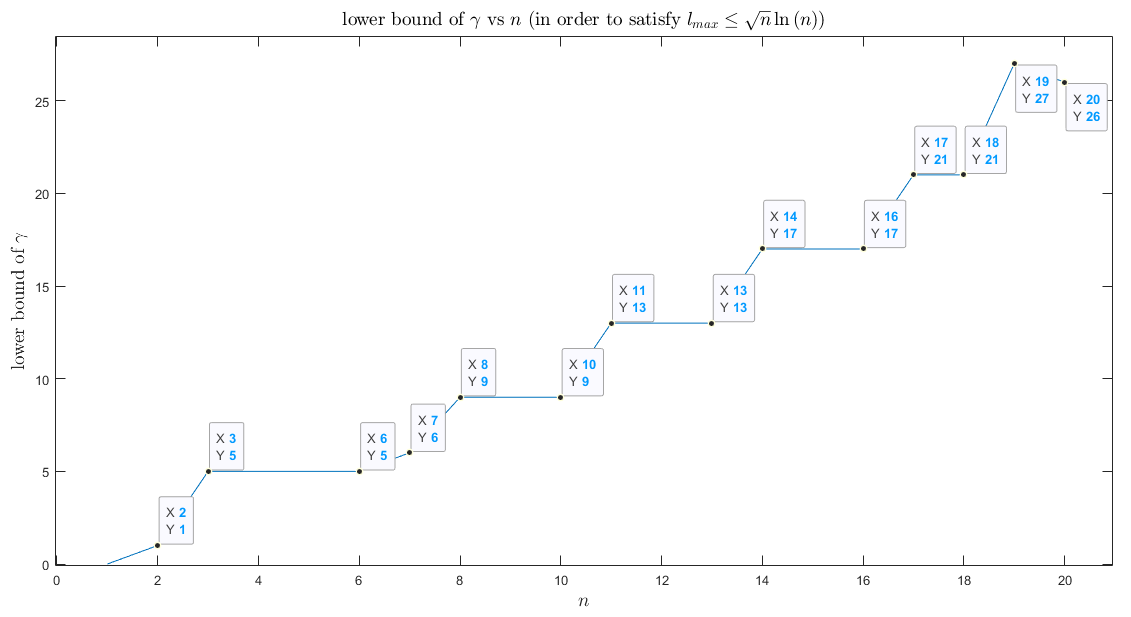}
		\caption{Lower bound of $\gamma$ in order to satisfy $l_{max}\leq\sqrt{n}\ln(n)$. For instance, point $(17,21)$  on the graph means that in order to generate all permutations on a 17-element set in less than $\sqrt{17}\ln(17)$ clock cycles, the generating set G must have at least 21 elements.}
		\label{fig:gamma_lower_bound}
	\end{figure}

	For sufficiently large $n$'s, an asymptotic analysis can be helpful. Inequality \ref{identitiy} holds if and only if $\gamma^{l_{max}+1}\geq(\gamma-1)n!+1$. That is
	\begin{equation} \label{l_max_lower_bound}
		l_{max} \geq \frac{\log_b((\gamma-1)n!+1)-\log_b(\gamma)}{\log_b(\gamma)}
	\end{equation}
	In the appendix, we have proved that for all $n\geq5$ and $\gamma\geq2$, we have
	\begin{equation} \label{appendix_ineq}
		\log_b((\gamma-1)n!+1)-\log_b(\gamma)>\frac{n}{2}\log_b(n)
	\end{equation}
	Therefore, using inequality \eqref{l_max_lower_bound}, we conclude that $l_{max} \geq \frac{\log_b((\gamma-1)n!+1)-\log_b(\gamma)}{\log_b(\gamma)}>\frac{\frac{n}{2}\log_b(n)}{\log_b(\gamma)}$. For sufficiently large $n$'s,
	\begin{equation} \label{l_max_LB}
		l_{max}>\frac{n\log_b(n)}{2\log_b(\gamma)}
	\end{equation}
	gives us an approximate, albeit sometimes optimistic, lower bound for the maximum length the permutations will have.
	
	For instance, in the case that $G$ is the set of transpositions, $\gamma = \left(\frac{n}{2}\right)$. Then for $n \geq 5$, $l_{\text{max}} > \frac{n \log_b(n)}{2 \log_b\left(\frac{n}{2}\right)} > \frac{n \log_b(n)}{2 \log_b(n^2)} = \frac{n}{4}$; i.e., the proposed hardware can't decrease the maximum length of permutations to $\frac{n}{4}$ or less.
	
	Nevertheless, the inequality \ref{l_max_LB} also helps us estimate a lower bound for $\gamma$ if we want to decrease the time complexity. According to inequality \ref{l_max_LB}, if we want to satisfy the condition $l_{\text{max}} \leq L$, it leads to $\frac{\frac{n}{2} \log_b(n)}{\log_b(\gamma)} < L$. Hence $\log_b(\gamma) > \frac{\frac{n}{2} \log_b(n)}{L}$. For instance, if $l_{\text{max}} \leq \sqrt{n} \log_b(n)$ where $n \geq 5$, we conclude that $\log_b(\gamma) > \frac{\frac{n}{2} \log_b(n)}{\sqrt{n} \log_b(n)} = \frac{\sqrt{n}}{2}$; i.e., $\gamma > b^{\frac{\sqrt{n}}{2}}$. Figure \ref{fig:lower_bounds} provides a comparison between the minimum $\gamma$ obtained from inequalities \ref{identitiy} and \ref{l_max_LB} for $n \leq 100$ and $b=e$. As you can see, inequality \ref{l_max_LB} provides a necessary condition for $\gamma$.
	
	\begin{figure}[h]
		\centering
		\includegraphics[scale=0.5]{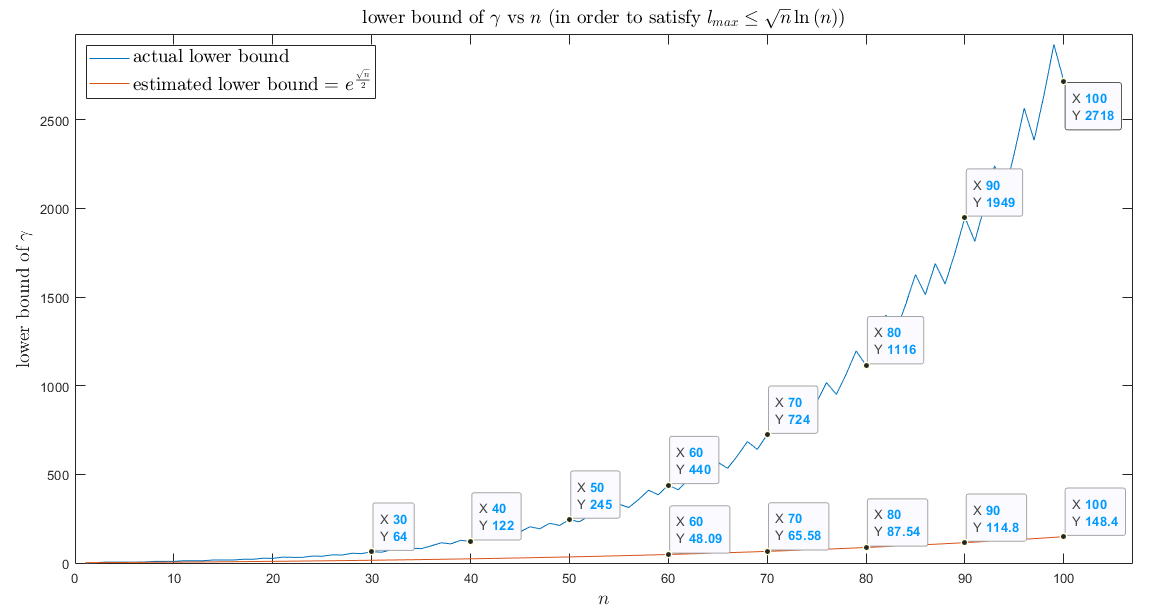}
		\caption{Actual and estimated lower bound of $\gamma$ in order to satisfy $l_{max}\leq\sqrt(n)\ln(n)$. This figure also shows that the estimated lower bound ($e^\frac{\sqrt{n}}{2}$) provides just a necessary but not sufficient condition for $\gamma$.}
		\label{fig:lower_bounds}
	\end{figure}
	
	The last remark we should consider is that in order to achieve a complete comparison and understanding of how the hardware methods work, it is crucial to implement (or at least simulate) them. One of the main hardware design principles is that \textquotedblleft Smaller is faster,\textquotedblright~which means the more complex the hardware, the slower it gets; since \textquotedblleft it takes electronic signals longer when they must travel farther.\textquotedblright~However, guidelines like this are not absolute. For instance, \textquotedblleft31 registers may not be faster than 32\textquotedblright\cite{Patterson2013-uy}. As a result, dedicating extra memory in order to lower the primitive operations does not always guarantee faster implementations since we may have to lower the clock frequency. This leads to another design principle: \textquotedblleft Good design demands good compromises\textquotedblright \cite{Patterson2013-uy}.

	\section{Conclusion}
	Random permutation generation (RPG) has a wide range of applications in computer science. In many applications, the size of the array we want to shuffle is fixed, and the shuffling process is done frequently. These applications made us try to speed up the procedure of RPG for arrays of a specific length. The well-known algorithm for this purpose is the Fisher-Yates algorithm. However, this algorithm sometimes wastes some clock cycles to do nothing. Our proposed hardware algorithm tries to avoid these wasted times.

	First of all, we provided a theoretical background. It was made up of five different insights: algebraic (when dealing with permutations as a group), language-theoretic (which provided an interface between algebraic and automatic insights), automatic (which provided a compact structure to store the information), machinelike (which was the closest insight to hardware design), and graphical (which was the interface between automatic and machinelike insights and helped us lower the amount of abstraction). In theorems \ref{minimal_DFA} and \ref{remarks}, we proved the minimality of the DFA and the length of words that the transducer produces. As a result, we have used the optimal solution in order to obtain an optimal method with respect to the number of needed transpositions.

	In section \ref{Hardware_Design}, we introduced a hardware design based on the theoretical background and proofs provided formerly. section \ref{comparing_performance} explained that why and how much our proposed method speeds up the RPG process compared with the Fisher-Yates algorithm. As we saw, the advantage of the proposed method would vanish as $n$ grew. However, for small $n$'s it has a significant speed-up. For $n\leq32$, the speed-up is at least 10.95\%, and for $n\leq85$, it is at least 5\%. It is the speed or hardware priorities that determine for which $n$'s it is cost-effective to use the proposed method. For example, for $n\leq605$, the speed-up is at least 1\%, but this amount may be too small for some applications. On the other hand, it may be significant in a data center. As we saw in section \ref{time_and_space}, our proposed method did not improve the time complexity. However, we generalized our method to contain sublinear time complexities too.

	We think that realizing the process we used was the most important point in this paper, which can pave the way for further research. There can be much research in the field of solving permutation puzzles optimally or representing algebraic problems as automatic ones or vice versa. Moreover, the implementation of proposed hardware or the hardware pieces of the generalized method is of great value since the implementation always involves many challenges and compromises. This paper was focused more on the theory. We hope we can complete its practical part through future research.

	
	%

	\appendices
	\section{Proof of the Inequality {\eqref{appendix_ineq}}} \label{app:appendixA}
	As needed in section \ref{time_and_space}, we want to prove that for all natural numbers $n \geq 5$ and $\gamma \geq 2$, and for all real numbers $b>1$, inequality $\log_b((\gamma-1)n!+1)-\log_b(\gamma)>\frac{n}{2}\log_b(n)$ holds. We prove this proposition through the following lemmas.
	
	\begin{lemma}
		For all natural numbers $n$ and $k$ such that $0 \leq k < n$, $\binom{n-1}{k} \leq n^k$.
	\end{lemma}
	\begin{proof}
		When $k=0$, the statement is obvious, since $\binom{n-1}{0}=1 \leq 1=n^0$. Suppose $k \neq 0$. Then
		\[
		\binom{n-1}{k} = \frac{(n-1)...(n-k)}{k!} = \frac{(n-1)}{1} \times \frac{(n-2)}{2} \times ... \times \frac{(n-k)}{k} < n^k
		\]
	\end{proof}
	
	\begin{corollary} \label{cor_multiplication}
		For all natural numbers $n$ and $k$ such that $0 \leq k < n$, $\binom{n-1}{k} n^{n-1-k} \leq n^{n-1}$
	\end{corollary}
	
	\begin{lemma} \label{beautiful_lemma}
		For all natural numbers $n$, $n^n \geq (n+1)^{n-1}$
	\end{lemma}
	\begin{proof}
		Using binomial expansion and corollary \ref{cor_multiplication}, the statement will be proved.
		\[
		(n+1)^{n-1} = \sum_{i=0}^{n-1} \binom{n-1}{i} n^{n-1-i} \leq \sum_{i=0}^{n-1} n^{n-1} = n(n^{n-1}) = n^n
		\]
	\end{proof}
	
	\begin{lemma} \label{n_geq_5}
		For all natural numbers $n \geq 5$, $\frac{n!}{2} > n^{\frac{n}{2}}$.
	\end{lemma}
	
	\begin{proof}
		We prove the statement by induction. For $n=5$, $\frac{n!}{2}=60 > 55.9 \approx n^{\frac{n}{2}}$. Now suppose the statement is true for $n=n_0$; i.e., $\frac{n_0!}{2} > n_0^{\frac{n_0}{2}}$. By multiplying both sides by $n_0+1$ we conclude $\frac{\left(n_0+1\right)!}{2} > n_0^{\frac{n_0}{2}}(n_0+1)$. It suffices to show the RHS is greater than or equal to $(n_0+1)^{\frac{n_0+1}{2}}$. This is true because
		\[n_0^{\frac{n_0}{2}}(n_0+1) \geq (n_0+1)^{\frac{n_0+1}{2}} \iff n_0^{\frac{n_0}{2}} \geq (n_0+1)^{\frac{n_0-1}{2}} \iff n_0^{n_0} \geq (n_0+1)^{n_0-1}\]
		The last inequality is just what Lemma \ref{beautiful_lemma} says.
	\end{proof}
	
	\begin{corollary} \label{semi_final}
		For all natural numbers $n \geq 5$, $\frac{n!+1}{2} > n^{\frac{n}{2}}$
	\end{corollary}
	
	\begin{theorem}
		For all natural numbers $n \geq 5$ and $\gamma \geq 2$, and for all real numbers $b>1$,  the inequality $\log_b((\gamma-1)n!+1)-\log_b(\gamma) > \frac{n}{2} \log_b(n)$ holds.
	\end{theorem}
	
	\begin{proof}
		Consider the bivariate function $f: \mathbb{N}^{\geq 5} \times \mathbb{N}^{\geq 2} \rightarrow \mathbb{R}$ with the function rule $f(n,\gamma) = (\gamma-1)n!+1-\gamma n^{\frac{n}{2}}$. We prove that for all $n$ and $\gamma$ in the domain of $f$, $f(n,\gamma) > 0$. First note that $f(n,2) = n!+1-2n^{\frac{n}{2}}$ and based on Corollary \ref{semi_final}, $f(n,2) > 0$. Furthermore, $f(n,\gamma+1)-f(n,\gamma) = n!-n^{\frac{n}{2}} = (n!-2n^{\frac{n}{2}})+n^{\frac{n}{2}}$ which is positive according to Lemma \ref{n_geq_5}. Therefore, the function $f$ is ascending with respect to $\gamma$. Hence, for all $\gamma > 2$, $f(n,\gamma) > f(n,2) > 0$. Finally, $(\gamma-1)n!+1-\gamma n^{\frac{n}{2}} > 0$ implies $\log_b((\gamma-1)n!+1)-\log_b(\gamma) > \frac{n}{2} \log_b(n)$.
	\end{proof}

	\ifCLASSOPTIONcaptionsoff
	\newpage
	\fi

	
	
	\bibliographystyle{IEEEtran}
	

\clearpage
	\begin{thebibliography}{33}
		\bibitem{Andreeva2015-di}
		E. Andreeva, B. Bilgin, A. Bogdanov, A. Luykx, B. Mennink, N. Mouha, and K. Yasuda, "{APE}: Authenticated {Permutation-Based} Encryption for Lightweight Cryptography," in \emph{Fast Software Encryption}, Springer Berlin Heidelberg, 2015, pp. 168--186.
		
		\bibitem{Wang2021-vz}
		J. Wang, X. Zhi, X. Chai, and Y. Lu, "Chaos-based image encryption strategy based on random number embedding and {DNA-level} self-adaptive permutation and diffusion," \emph{Multimed. Tools Appl.}, vol. 80, no. 10, pp. 16087--16122, Apr. 2021.
		
		\bibitem{Punithavathi2021-sg}
		P. Punithavathi and S. Geetha, "Random {Permutation-Based} Linear Discriminant Analysis for Cancelable Biometric Recognition," in \emph{Advances in Computing and Network Communications}, Springer Singapore, pp. 593--603, 2021.
		
		\bibitem{Zheng2022-im}
		F. Zheng, C. Chen, X. Zheng, and M. Zhu, "Towards secure and practical machine learning via secret sharing and random permutation," \emph{Knowledge-Based Systems}, vol. 245, pp. 108609, Jun. 2022.
		
		\bibitem{Kao1996-yo}
		M. Kao, J. H. Reif, and S. R. Tate, "Searching in an Unknown Environment: An Optimal Randomized Algorithm for the {Cow-Path} Problem," \emph{Inform. and Comput.}, vol. 131, no. 1, pp. 63--79, Nov. 1996.
		
		\bibitem{Motwani1995-ww}
		R. Motwani and P. Raghavan, \emph{Randomized algorithms}, Cambridge University Press, 1995.
		
		\bibitem{Hemerik2018-ue}
		J. Hemerik and J. Goeman, "Exact testing with random permutations," \emph{Test}, vol. 27, no. 4, pp. 811--825, 2018.
		
		\bibitem{Berry2014-hp}
		K. J. Berry, J. E. Johnston, and P. W. Mielke, Jr., \emph{A Chronicle of Permutation Statistical Methods: 1920--2000, and Beyond}, Springer, 2014.
		
		\bibitem{Li2013-la}
		R. Li, M. Wang, L. Jin, and Y. He, "A Monte Carlo permutation test for random mating using genome sequences," \emph{PLoS One}, vol. 8, no. 8, e71496, Aug. 2013.
		
		\bibitem{Manly2018-uf}
		B. F. J. Manly, \emph{Randomization, bootstrap and Monte Carlo methods in biology}, chapman and hall/CRC, 2018.
		
		\bibitem{Mishchenko2020-ei}
		K. Mishchenko, A. Khaled, and P. Richtárik, "Random reshuffling: Simple analysis with vast improvements," in \emph{Proc. 34th Conf. Neural Inf. Process. Syst. (NeurIPS)}, Vancouver, Canada, 2020.
		
		\bibitem{Gan2010-ya}
		L. Gan, T. T. Do, and T. D. Tran, "Fast dimension reduction through random permutation," in \emph{2010 IEEE International Conference on Image Processing}, pp. 3353--3356, Sep. 2010.
		
		\bibitem{Fisher1938-aq}
		R. A. Fisher and F. Yates, \emph{Statistical tables for biological, agricultural and medical research}, Oliver and Boyd, 1938.
		
		\bibitem{Durstenfeld1964-yr}
		R. Durstenfeld, "Algorithm 235: Random permutation," \emph{Commun. ACM}, vol. 7, no. 7, p. 420, Jul. 1964.
		
		\bibitem{Knuth1998-mz}
		G. Knuth, "The art of computer programming, seminumerical algorithms, vol. 2, addition wesley," \emph{Reading, Massachusetts}, 1998.
		
		\bibitem{OConnor2014-bq}
		D. O'Connor, "A Historical Note on Shuffle Algorithms," \emph{Retrieved Maret}, vol. 4, pp. 2018, 2014.
		
		\bibitem{Arndt2010-nt}
		J. Arndt, "Generating random permutations," Ph.D. dissertation, Australian National University, Mar. 2010.
		
		\bibitem{Odom2019-co}
		J. H. Odom, "Indexing Large Permutations in Hardware," M.S. thesis, Virginia Polytechnic Institute and State University, 2019.
		
		\bibitem{Magnus2004-in}
		W. Magnus, A. Karrass, and D. Solitar, \emph{Combinatorial group theory: Presentations of groups in terms of generators and relations}, Courier Corporation, 2004.
		
		\bibitem{Vidal2005-xv}
		E. Vidal, F. Thollard, C. de la Higuera, F. Casacuberta, and R. C. Carrasco, "Probabilistic finite-state machines--part {I}," \emph{IEEE Trans. Pattern Anal. Mach. Intell.}, vol. 27, no. 7, pp. 1013--1025, Jul. 2005.
		
		\bibitem{Malik1997-mj}
		D. S. Malik, J. M. Mordeson, and M. K. Sen, \emph{Fundamentals of Abstract Algebra}, McGraw-Hill, 1997.
		
		\bibitem{Linz2017-vb}
		P. Linz, \emph{An Introduction to Formal Languages and Automata}, Jones \& Bartlett Learning, 2017.
		
		\bibitem{Joyner2008-gs}
		D. Joyner, \emph{Adventures in Group Theory: Rubik's Cube, Merlin's Machine, and Other Mathematical Toys}, JHU Press, Dec. 2008.
		
		\bibitem{Holcombe2004-lb}
		W. M. L. Holcombe, \emph{Algebraic Automata Theory}, Cambridge University Press, Jun. 2004.
		
		\bibitem{Godin2017-rd}
		T. Godin, "An analogue to Dixon's theorem for automaton groups," in \emph{2017 Proceedings of the Meeting on Analytic Algorithmics and Combinatorics (ANALCO)}, Society for Industrial and Applied Mathematics, Jan. 2017.
		
		\bibitem{Harington1897-oh}
		A. Harington, "{ANIMAL} {AUTOMATISM} {AND} {CONSCIOUSNESS}," \emph{Monist}, vol. 7, no. 4, pp. 611--616, 1897.
		
		\bibitem{Gradel2020-ur}
		E. Grädel, "Automatic Structures: Twenty Years Later," in \emph{Proceedings of the 35th Annual ACM/IEEE Symposium on Logic in Computer Science (LICS '20)}, Association for Computing Machinery, pp. 21--34, Jul. 2020.
		
		\bibitem{WikiDiff2018-hy}
		WikiDiff, "Automatic - What does it mean?," \emph{WikiDiff}, Apr. 2018. [Online]. Available: \url{https://wikidiff.com/automatic}. [Accessed: Nov. 27, 2023].
		
		\bibitem{Lee2017-rl}
		E. A. Lee and S. A. Seshia, \emph{Introduction to embedded systems: A cyber-physical systems approach}, MIT Press, 2017.
		
		\bibitem{Fialkow1992-sf}
		L. Fialkow and H. Salas, "Data Exchange and Permutation Length," \emph{Math. Mag.}, vol. 65, no. 3, pp. 188--193, Jun. 1992.
		
		\bibitem{LIPOWSKI20122193}
		A. Lipowski and D. Lipowska, "Roulette-wheel selection via stochastic acceptance," \emph{Physica A: Statistical Mechanics and its Applications}, vol. 391, no. 6, pp. 2193--2196, 2012.
		
		\bibitem{Graham1989-tb}
		R. L. Graham, D. E. Knuth, O. Patashnik, and S. Liu, "Concrete mathematics: a foundation for computer science," \emph{Computers in Physics}, aip.scitation.org, 1989.
		
		\bibitem{Patterson2013-uy}
		D. A. Patterson and J. L. Hennessy, \emph{Computer Organization and Design {MIPS} Edition: The {Hardware/Software} Interface}, Morgan Kaufmann, 2013.
	\end{thebibliography}
	%
	
	%
	
	
	

\end{document}